\documentclass[reqno]{amsart}
\usepackage{amsmath,amsthm,amssymb,eucal,fullpage,mathrsfs,setspace,bm,mathabx}
\usepackage{graphicx}
\usepackage[round]{natbib}
\usepackage[all,cmtip]{xy}
\usepackage{lineno}
\usepackage{subfig}

\title{stochastic selection processes}
\author{alex mcavoy}

\theoremstyle{definition}

\newtheorem{definition}{Definition}
\newtheorem{example}{Example}

\newtheorem{proposition}{Proposition}
\newtheorem{remark}{Remark}

\begin{document}

\begin{abstract}
We propose a mathematical framework for natural selection in finite populations. Traditionally, many of the selection-based processes used to describe cultural and genetic evolution (such as imitation and birth-death models) have been studied on a case-by-case basis. Over time, these models have grown in sophistication to include population structure, differing phenotypes, and various forms of interaction asymmetry, among other features. Furthermore, many processes inspired by natural selection, such as evolutionary algorithms in computer science, possess characteristics that should fall within the realm of a ``selection process," but so far there is no overarching theory encompassing these evolutionary processes. The framework of \textit{stochastic selection processes} we present here provides such a theory and consists of three main components: a \textit{population state space}, an \textit{aggregate payoff function}, and an \textit{update rule}. A population state space is a generalization of the notion of population structure, and it can include non-spatial information such as strategy-mutation rates and phenotypes. An aggregate payoff function allows one to generically talk about the fitness of traits without explicitly specifying a method of payoff accounting or even the nature of the interactions that determine payoff/fitness. An update rule is a fitness-based function that updates a population based on its current state, and it includes as special cases the classical update mechanisms (Moran, Wright-Fisher, etc.) as well as more complicated mechanisms involving chromosomal crossover, mutation, and even complex cultural syntheses of strategies of neighboring individuals. Our framework covers models with variable population size as well as with arbitrary, measurable trait spaces.
\end{abstract}

\maketitle

\allowdisplaybreaks

\section{Introduction}

Evolutionary game theory has proven itself extremely useful for modeling both cultural and genetic evolution \citep{maynardsmith:CUP:1982,hofbauer:CUP:1998,dugatkin:OUP:2000,nowak:BP:2006}. Traits, which are represented as strategies, determine the fitness of the players. An individual's fitness might be determined solely by his or her strategy (frequency-independent fitness) or it might also depend on the traits of the other players in the population (frequency-dependent fitness). The population evolves via a fitness-based update mechanism, and the long-run behavior of this process can be studied to determine which traits are more successful than others.

Evolutionary game theory was first used to study evolution in infinite populations via deterministic replicator dynamics \citep{taylor:MB:1978}. More recently, the dynamics of evolutionary games have also been studied in finite populations \citep{nowak:Nature:2004,taylor:BMB:2004}. Finite-population evolutionary dynamics typically have two timescales: one for interactions and one for updates. A player has a strategy (trait) and interacts with his or her neighbors in order to receive a payoff. In a birth-death process, for instance, this payoff is converted to reproductive fitness and used to update the population as follows: First, a player is chosen from the population for reproduction with probability proportional to (relative) fitness. Next, a player is chosen uniformly at random from the population for death, and the offspring of the reproducing player replaces the deceased player. This birth-death process is a frequency-dependent version of the classical Moran process \citep{moran:MPCPS:1958,nowak:BP:2006}.

Whereas replicator dynamics are deterministic, finite-population models of evolution are inherently stochastic and incorporate principles of both natural selection and genetic drift. All biological populations are finite, and we focus here on stochastic evolutionary games in finite populations. In particular, we focus on \textit{selection} processes, which, informally, are evolutionary processes in which the update step depends on fitness. Selection processes typically resemble birth-death processes in that there is reproduction and replacement, with ``fitter" players more likely to reproduce than those with lower fitness. Of course, the order and number of births and deaths may vary, and the update might instead be based on imitation instead of reproduction. One of our goals here is to precisely define \textit{selection process} in a way that captures all of the salient features of the classical models of evolution in finite populations.

\citet{nowak:PTRSB:2009} state that ``There is (as yet) no general mathematical framework that would encompass evolutionary dynamics for any kind of population structure." We propose here such a framework, which we term \textit{stochastic selection processes}, to describe the evolutionary games used to model natural selection in finite populations. Stochastic selection processes model evolutionary processes with two timescales: one for interactions (which determine fitness) and one for updates (selection, mutation, etc.). Our framework takes into account arbitrary population structures, as well as non-spatial information about the population such as phenotypes and strategy mutations. Moreover, this framework encompasses all types of strategy spaces, games (matrix, asymmetric, multiplayer, etc.), and fitness-based update rules.

An example of ambiguity in evolutionary game theory is that classical games, such as two-player matrix games, are often used to define evolutionary processes in populations, and in this context the term ``game" can refer to either the classical game or to the evolutionary process. Moreover, classical multiplayer games, such as the public goods game, can result in processes in which each player in the population derives a payoff from \textit{several} multiplayer interactions, and some of these interactions might involve players who are not neighbors \cite[see][]{mcavoy:JMB:2015}. Even further, when a player is involved in multiple interactions, the total payoff to this player may be derived in more than one way: payoffs from individual encounters may be accumulated (added) or averaged, for instance, and the nature of this accounting can strongly influence the evolutionary process \citep{maciejewski:PLoSCB:2014}.

In order to accommodate the many ways of deriving ``total payoff" from a classical game, one may distill from these methods a common feature: each player has a strategy and receives an aggregate payoff from a sequence of interactions. That is, if $S$ is the strategy space available to each of the $N$ players in the population, then there is an \textit{aggregate payoff function}, $u:S^{N}\rightarrow\mathbb{R}^{N}$, such that the $i$th coordinate function, $u_{i}$, is the payoff to player $i$ for all of the (microscopic) interactions in which this player is involved. As an example, consider the two-player game defined by the matrix
\begin{linenomath}
\begin{align}
\bordermatrix{%
 & A & B \cr
A &\ a & \ b \cr
B &\ c & \ d \cr
} .
\end{align}
\end{linenomath}
Suppose that the population is well mixed so that each player interacts with every other player. Let $S=\left\{A,B\right\}$ and let $s\in S^{N}$ be a strategy profile consisting of $k+1$ players using $A$ and $N-1-k$ players using $B$ (the ordering is not important since the population is well mixed). If player $i$ is an $A$-player, then the payoff to this player is
\begin{linenomath}
\begin{align}
u_{i}^{\textrm{acc}}\left(s\right) &:= ka+\left(N-1-k\right) b
\end{align}
\end{linenomath}
if payoffs are accumulated and
\begin{linenomath}
\begin{align}
u_{i}^{\textrm{ave}}\left(s\right) &:= \frac{ka+\left(N-1-k\right) b}{N-1}
\end{align}
\end{linenomath}
if payoffs are averaged. These two methods of calculating payoffs from pairwise interactions give essentially equivalent evolutionary dynamics since the population is well mixed, but this phenomenon need not hold for more complicated methods of payoff accounting or (spatial) population structures. Evolutionary dynamics aside, this example illustrates how one can define an aggregate payoff function, $u$, from a classical game such as a $2\times 2$ matrix game; different methods of obtaining total payoff from a series of interactions result in different aggregate payoff functions.

A selection process in a finite population typically has for a state space the set of all strategy profiles, i.e. the set of all $N$-tuples of strategies \citep{allen:JMB:2012}. A strategy profile indicates a strategy for each player in the population, and often an evolutionary process updates only these strategies. The population structure may be fixed \citep{lieberman:Nature:2005,szabo:PR:2007} or dynamic \citep{tarnita:PNAS:2009,wardil:SR:2014}. (In the latter case, one must also account for the population structure in the state space of the process.) An aggregate payoff function, which takes into account both strategies and population structure, assigns a payoff to each player in the population. The payoff to player $i$, $u_{i}$, is then converted to fitness, $f_{i}=f\left(u_{i}\right)$, where $f$ is some payoff-to-fitness function (e.g. $f\left(u_{i}\right) =\exp\left\{\beta u_{i}\right\}$ for some $\beta\geqslant 0$). A fitness-based update rule such as birth-death \citep{moran:MPCPS:1958,nowak:Nature:2004}; death-birth or imitation \citep{ohtsuki:Nature:2006,ohtsuki:JTB:2006}; pairwise comparison \citep{szabo:PRE:1998,traulsen:JTB:2007}; or Wright-Fisher \citep{ewens:S:2004,imhof:JMB:2006} is then repeatedly applied to the population, at each step updating the strategies of the players and (possibly) the population structure.

Based on this pattern, it seems reasonable to define the state of an evolutionary process to be a pair, $\left(s,\Gamma\right)$, where $s$ is an $N$-tuple of strategies and $\Gamma$ is a population structure. However, there may be more to the state of a population than its spatial structure. For example, \citet{antal:PNAS:2009} consider evolution in \textit{phenotype space}, a model in which each player has both a strategy and a phenotype, and phenotypes influence the effects of strategies on interactions. In the Prisoner's Dilemma, for instance, a cooperator (whose strategy is $C$) cooperates only with other players who are phenotypically similar and defects otherwise. In this instance, we show how one can consider phenotypes as a part of the ``population state" in the sense that they contain information about the players in the population. In general, a state of the evolutionary process can be represented by a pair, $\left(s,\mathscr{P}\right)$, where $s$ is a strategy profile and $\mathscr{P}$ is a \textit{population state} (which we make mathematically precise). Notably, the population state is distinct from the strategies of the players; it describes all of the non-strategy information about the players.

When viewed from this perspective, the effects of phenotypes on strategies can be implemented directly in the aggregate payoff function, $u$: a player facing a cooperator receives the payoff for facing a cooperator if and only if they are phenotypically similar. Similarly, strategy mutations can also be considered as a part of the population state and accounted for directly in the update step of the process. As a part of our analysis, we formally define \textit{aggregate payoff function} and \textit{update rule} and show how they are influenced by the components of the population state.

This setup, which involves a series of interactions as the population transitions through various states, is reminiscent of a \textit{stochastic game} \citep{shapley:PNAS:1953}. A stochastic game is played in stages, with each stage consisting of a normal-form game determined by some ``state." The game played in the subsequent stage is determined probabilistically by the current state as well as the strategies played in the current stage. We will see that, in general, selection processes are not necessarily stochastic games, and neither are stochastic games necessarily selection processes. However, these two types of processes do share some common features, and we use some of the components of a stochastic game as inspiration for our framework.

Many problems in evolution, such as how and why cooperation evolves, depend on the specifics of the update rule, population structure, mutation rates, etc. \citep{ohtsuki:Nature:2006,taylor:Nature:2007,traulsen:PNAS:2009,debarre:NC:2014,rand:PNAS:2014}. We clarify how these pieces (among others) fit together. Our objective here is threefold: (1) to compare and contrast evolutionary and stochastic games, drawing inspiration from the latter to describe the former; (2) to propose a general mathematical framework encompassing natural selection models in finite populations; and (3) to examine the components of several classical evolutionary processes and demonstrate how they fit into our framework. Many (if not most) of the existing models of evolution in finite populations involve a fixed population size. Therefore, our framework is first stated in terms of a fixed population size since this setting most readily allows for a comparison to the theory of stochastic games and for illustrative examples placing several standard evolutionary processes into a broader context. However, the assumption that the population size is fixed is not crucial to our theory, and we conclude by extending our framework to processes with variable population size.

\section{Stochastic games}

Prior to outlining the basic theory of stochastic games, we first need to recall some definitions and notation. A \textit{measurable space} consists of a set, $X$, and a $\sigma$-algebra of sets, $\mathcal{F}\left(X\right)$, on $X$. Often, we refer to $X$ itself as a ``measurable space" and suppress $\mathcal{F}\left(X\right)$. If $X$ is a measurable space, then we denote by $\Delta\left(X\right)$ the space of probability measures on $X$; that is, if $\mathcal{M}\left(X\right)$ is the space of all measures on $\left(X,\mathcal{F}\left(X\right)\right)$, then
\begin{linenomath}
\begin{align}
\Delta\left(X\right) &:= \left\{ \mu\in\mathcal{M}\left(X\right)\ :\ \mu\left(X\right) = 1 \right\} .
\end{align}
\end{linenomath}
For measurable spaces $X$ and $Y$, denote by $K\left(X,Y\right)$ the set of Markov kernels from $X$ to $Y$; that is, $K\left(X,Y\right)$ is the set of functions $\kappa :X\times\mathcal{F}\left(Y\right)\rightarrow\left[0,1\right]$ such that (i) $\kappa\left(x,\--\right) :\mathcal{F}\left(Y\right)\rightarrow\left[0,1\right]$ is in $\Delta\left(Y\right)$ for each $x\in X$ and (ii) $\kappa\left(\-- ,E\right) :X\rightarrow\left[0,1\right]$ is measurable for each $E\in\mathcal{F}\left(Y\right)$. We also write $\kappa :X\rightarrow\Delta\left(Y\right)$ to denote such a kernel.

\citet{shapley:PNAS:1953} considers a collection of normal-form games, together with a probabilistic rule for transitioning between these games, which he refers to as a \textit{stochastic game}. A stochastic game is a generalization of a repeated game that, formally, consists of the following components:
\begin{enumerate}

\item[(i)] $N$ players, labeled $1,\dots ,N$;

\item[(ii)] for each player, $i$, a measurable strategy space, $S_{i}$;

\item[(iii)] a measurable state space, $\mathcal{P}$;

\item[(iv)] a ``single-period" payoff function, $u:\mathbf{S}\times\mathcal{P}\rightarrow\mathbb{R}^{N}$, where $u_{i}$ is the payoff to player $i$ and $\mathbf{S}:=S_{1}\times\cdots\times S_{N}$ is the set of all \textit{strategy profiles};

\item[(v)] a transition kernel, $T:\mathbf{S}\times\mathcal{P}\rightarrow\Delta\left(\mathcal{P}\right)$.

\end{enumerate}

A Markov decision process is a stochastic game with one player \citep{puterman:JWS:1994,neyman:FMCSG:2003}, and a repeated game is a stochastic game whose state space, $\mathcal{P}$, consists of just a single element \citep{mcmillan:R:2001,mertens:CUP:2015}.

Examples of strategies for stochastic games are the following \citep[see][]{neyman:S:2003}:
\begin{enumerate}

\item[(1)] \textit{Pure strategies}: Let $\mathcal{H}$ denote the set of all possible histories, i.e.
\begin{linenomath}
\begin{align}
\mathcal{H} &:= \left\{\varnothing\right\}\cup\bigsqcup_{t\geqslant 1}\left(\mathbf{S}\times\mathcal{P}\right)^{t} ,
\end{align}
\end{linenomath}
where $\varnothing$ denotes the ``null" history and $\bigsqcup_{t\geqslant 1}\left(\mathbf{S}\times\mathcal{P}\right)^{t}$ is the disjoint union of the spaces of $t$-tuples, $\left(\mathbf{S}\times\mathcal{P}\right)^{t}$, for $t\geqslant 1$. A pure strategy for player $i$ is a map
\begin{linenomath}
\begin{align}
s_{i} &: \mathcal{H} \longrightarrow S_{i} ,
\end{align}
\end{linenomath}
indicating an action in $S_{i}$ for each $t$ and history $h^{t}\in\left(\mathbf{S}\times\mathcal{P}\right)^{t}\subseteq\mathcal{H}$. We denote by $\textrm{Map}\left(\mathcal{H},S_{i}\right)$ the set of all such maps, i.e. the set of player $i$'s pure strategies.

\item[(2)] \textit{Mixed strategies:} A mixed strategy for player $i$ is a probability distribution over the set of pure strategies for player $i$, i.e. an element $\sigma_{i}\in\Delta\left(\textrm{Map}\left(\mathcal{H},S_{i}\right)\right)$.

\item[(3)] \textit{Behavioral strategies}: A behavioral strategy for player $i$ is a map
\begin{linenomath}
\begin{align}
\sigma_{i} &: \mathcal{H} \longrightarrow \Delta\left(S_{i}\right) ,
\end{align}
\end{linenomath}
indicating a distribution over $S_{i}$ for each $t$ and history $h^{t}\in\left(\mathbf{S}\times\mathcal{P}\right)^{t}\subseteq\mathcal{H}$.

\item[(4)] \textit{Markov strategies}: A Markov strategy for player $i$ is a behavioral strategy, $\sigma_{i}$, such that $\sigma_{i}\left(h^{t}\right) =\sigma_{i}\left(k^{t}\right)$ for each $t$ and $h^{t},k^{t}\in\left(\mathbf{S}\times\mathcal{P}\right)^{t}\subseteq\mathcal{H}$ with $h_{t}^{t}=k_{t}^{t}$. In other words, a Markov strategy is a ``memory-one" behavioral strategy, i.e. a behavioral strategy that depends on only the last strategy profile, state, and $t$.

\item[(5)] \textit{Stationary strategies}: A stationary strategy is a Markov strategy that is independent of $t$, i.e. a behavioral strategy that depends on only the last strategy profile and state.

\end{enumerate}

Of these five classes of strategies, behavioral strategies are the most general. Indeed, pure, mixed, Markov, and stationary strategies are all instances of behavioral strategies. In the context of repeated games, a memory-one strategy of the repeated game \citep{press:PNAS:2012} is equivalent to a stationary strategy of the stochastic game, and a longer-memory strategy of the repeated game \citep{hauert:PRSLB:1997} is equivalent to a behavioral strategy.

\subsection{Evolutionary processes as stochastic games}

At first glance, stochastic games seem to provide a reasonable framework for evolutionary games: a stochastic game transitions through states in stages (``periods"), which could be population structures or states, and in each stage the players receive payoffs based on a single-period payoff function, $u$. However, it is in the dynamics that the differences between stochastic and evolutionary games become evident:

The combination of a stochastic game and a strategy for each player defines a stochastic process on $\mathbf{S}\times\mathcal{P}$, although this stochastic process might or might not be a Markov chain. For example, let $T$ be the transition kernel for a stochastic game, and suppose that $\sigma_{i}$ is a stationary strategy for player $i$ for $i=1,\dots ,N$. Let $\kappa$ be the transition kernel on $\mathbf{S}\times\mathcal{P}$ defined by the product measure, i.e.
\begin{linenomath}
\begin{align}\label{eq:chainFromSG}
\kappa &: \mathbf{S}\times\mathcal{P} \longrightarrow \Delta\left(\mathbf{S}\times\mathcal{P}\right) \nonumber \\
&: \left(s,\mathscr{P}\right) \longmapsto \sigma_{1}\left[s,\mathscr{P}\right]\times\cdots\times\sigma_{N}\left[s,\mathscr{P}\right]\times T\left[s,\mathscr{P}\right] .
\end{align}
\end{linenomath}
Thus, the stochastic game--together with this profile of stationary strategies--defines a time-homogeneous Markov chain on $\mathbf{S}\times\mathcal{P}$. In general, if these strategies are instead Markov strategies, then the resulting Markov chain might be time-inhomogeneous. If these strategies are pure, mixed, or behavioral, then the stochastic process on $\mathbf{S}\times\mathcal{P}$ defined by the game need not even be a Markov chain.

Evolutionary processes are typically defined as Markov chains on $\mathbf{S}\times\mathcal{P}$, where $\mathcal{P}$ is chosen to be a state space appropriate for the evolutionary process (such as the space of population structures, mutation-rate profiles, or phenotype profiles). In light of the previous remarks, it is not unreasonable to expect that many of these processes are equivalent to stochastic games combined with stationary strategies. As it turns out, evolutionary processes generally possess correlations between updates of strategies and population states, forbidding an equivalence between an evolutionary process and a Markov chain constructed from a stochastic game via Eq. (\ref{eq:chainFromSG}). These correlations are evident already in one of the most basic models of evolution in finite populations:
\begin{example}[Moran process]\label{ex:moranProcess}
Suppose that $S=\left\{A,B\right\}$, where strategy $A$ represents the mutant type and strategy $B$ represents the wild type. A mutant type has fitness $r>0$ relative to the wild type (whose fitness relative to itself is just $1$). Let $\mathfrak{m}$ be a finite subset of $\left[0,1\right]$ consisting of a number of ``mutation rates," and let $\mathcal{P}:=\mathfrak{m}^{N}$. In a well-mixed population of size $N$, the Moran process proceeds as follows: In each time step, an individual (``player") is chosen for reproduction with probability proportional to relative fitness. Another player (including the one chosen for reproduction) is then chosen uniformly at random from the population for death. The offspring of the player chosen for reproduction replaces the deceased player. If player $i$ is chosen for reproduction and $\varepsilon_{i}\in\mathfrak{m}$ is this player's mutation rate, then the offspring of this player inherits the type of the parent with probability $1-\varepsilon_{i}$ and takes on a type uniformly at random from $\left\{A,B\right\}$ with probability $\varepsilon_{i}$. The mutation rate of the offspring is $\varepsilon_{i}$, which is inherited directly from the parent. Thus, a state of this process consists of a profile of types (``strategies"), $s\in S^{N}$, and a profile of mutation rates, $\varepsilon\in\mathcal{P}=\mathfrak{m}^{N}$.

A transition between states $\left(s,\varepsilon\right)$ and $\left(s',\varepsilon '\right)$ is possible only if there exists $j$ such that $s_{\ell}=s_{\ell}'$ and $\varepsilon_{\ell}=\varepsilon_{\ell}'$ for each $\ell\neq j$. If player $i$ is selected for reproduction and player $j$ is chosen for death, then it must be the case that $\varepsilon_{i}=\varepsilon_{j}'$. If $\delta_{s,t}=1$ when $s=t$ (and is $0$ otherwise), then the probability that player $i$ is selected for reproduction is
\begin{linenomath}
\begin{align}
\frac{r\delta_{s_{i},A}+\delta_{s_{i},B}}{\sum_{\ell =1}^{N}\left(r\delta_{s_{\ell},A}+\delta_{s_{\ell},B}\right)} .
\end{align}
\end{linenomath}
The probability that player $j$ is chosen for death is $1/N$. If the offspring of player $i$ inherits the strategy of the parent, then it must be true that $s_{i}=s_{j}'$. Otherwise, the offspring of player $i$ ``mutates" and adopts strategy $s_{j}'$ with probability $1/2$. Therefore, the probability of transitioning between states $\left(s,\varepsilon\right)$ and $\left(s',\varepsilon '\right)$ is
\begin{linenomath}
\begin{align}
\mathbf{T}_{\left(s,\varepsilon\right) ,\left(s',\varepsilon '\right)} &= \sum_{i,j=1}^{N}\left(\prod_{\ell\neq j}\delta_{s_{\ell},s_{\ell}'}\delta_{\varepsilon_{\ell},\varepsilon_{\ell}'}\right)\delta_{\varepsilon_{i},\varepsilon_{j}'} \nonumber \\
&\quad\quad\times \left(\frac{r\delta_{s_{i},A}+\delta_{s_{i},B}}{\sum_{\ell =1}^{N}\left(r\delta_{s_{\ell},A}+\delta_{s_{\ell},B}\right)}\right)\left(\frac{1}{N}\right)\left[\left(1-\varepsilon_{i}\right)\delta_{s_{i},s_{j}'}+\varepsilon_{i}\left(\frac{1}{2}\right)\right] . \label{eq:moranTransition}
\end{align}
\end{linenomath}
By Eq. (\ref{eq:moranTransition}), the distributions on $S^{N}$ and $\mathcal{P}$, respectively, are \textit{not} independent.
\end{example}

More formally, consider a Markov chain on $\mathbf{S}\times\mathcal{P}$ defined by some evolutionary process such as the Moran process of Example \ref{ex:moranProcess}, and let $\kappa$ be its transition kernel. The projection maps $\Pi_{1}:\mathbf{S}\times\mathcal{P}\rightarrow\mathbf{S}$ and $\Pi_{2}:\mathbf{S}\times\mathcal{P}\rightarrow\mathcal{P}$ produce pushforward maps
\begin{linenomath}
\begin{subequations}
\begin{align}
\left(\Pi_{1}\right)_{\ast} &: \Delta\left(\mathbf{S}\times\mathcal{P}\right) \longrightarrow \Delta\left(\mathbf{S}\right) \nonumber \\
&: \mu \longmapsto \mu\circ\Pi_{1}^{-1} ; \label{firstPushforward} \\
\left(\Pi_{2}\right)_{\ast} &: \Delta\left(\mathbf{S}\times\mathcal{P}\right) \longrightarrow \Delta\left(\mathcal{P}\right) \nonumber \\
&: \mu \longmapsto \mu\circ\Pi_{2}^{-1} . \label{secondPushforward}
\end{align}
\end{subequations}
\end{linenomath}
From $\kappa$, we obtain a transition kernel for a stochastic game, $T$, defined by
\begin{linenomath}
\begin{align}
T\left[s,\mathscr{P}\right] &:= \left(\Pi_{2}\right)_{\ast}\kappa\left[s,\mathscr{P}\right]
\end{align}
\end{linenomath}
for each $s\in\mathbf{S}$ and $\mathscr{P}\in\mathcal{P}$. Similarly, we obtain the (stationary) strategy profile
\begin{linenomath}
\begin{align}
\sigma\left[s,\mathscr{P}\right] &:= \left(\Pi_{1}\right)_{\ast}\kappa\left[s,\mathscr{P}\right] .
\end{align}
\end{linenomath}
However, one typically loses information in passing from such a Markov chain to the combination of a stochastic game and a profile of stationary strategies. First of all, the transition kernel $\kappa$ generally cannot be reconstructed from $T$ and $\sigma$ since $\kappa\left[s,\mathscr{P}\right]$ need not be in $\Delta\left(\mathbf{S}\right)\times\Delta\left(\mathcal{P}\right)$; a priori, we know only that $\kappa\left[s,\mathscr{P}\right]\in\Delta\left(\mathbf{S}\times\mathcal{P}\right)$ for $s\in\mathbf{S}$ and $\mathscr{P}\in\mathcal{P}$ (see Eq. (\ref{eq:moranTransition}), for example). Moreover, $\sigma$ need not be of the form $\left(\sigma_{1},\dots ,\sigma_{N}\right)$ for stationary strategies $\sigma_{1},\dots ,\sigma_{N}$; in particular, $\sigma$ is a \textit{correlated} stationary profile \citep[see][]{aumann:E:1987,fudenberg:MIT:1991}. In other words, whereas a sequence of independent strategy choices produce an element
\begin{linenomath}
\begin{align}
\Big(\sigma_{1}\left[s,\mathscr{P}\right],\dots ,\sigma_{N}\left[s,\mathscr{P}\right]\Big)\in\Delta\left(S_{1}\right)\times\cdots\times\Delta\left(S_{N}\right) ,
\end{align}
\end{linenomath}
it might be the case that $\sigma\left[s,\mathscr{P}\right]\in\Delta\left(S_{1}\times\cdots\times S_{N}\right)\--\Big(\Delta\left(S_{1}\right)\times\cdots\times\Delta\left(S_{N}\right)\Big)$.

In \S\ref{sec:ssp}, we present a framework for stochastic evolutionary processes used to model natural selection. These processes, which we call \textit{stochastic selection processes}, illustrate more clearly the correlations between distributions on $\mathbf{S}$ and $\mathcal{P}$ arising in many evolutionary processes. We saw in Example \ref{ex:moranProcess} that, despite the similarities between stochastic games and evolutionary processes, there are important differences between the two frameworks. However, our notion of a stochastic selection process draws inspiration from the theory of stochastic games. Namely, we appropriate the concepts of (i) state space, $\mathcal{P}$; (ii) single-period payoff function, $u$; and (iii) update step.

\section{Stochastic selection processes with fixed population size}\label{sec:ssp}

Here we focus on a type of evolutionary process that we term a \textit{stochastic selection process}. Stochastic selection processes seek to model processes with two timescales: one for interactions and one for selection. In the interaction step, players interact with one another and receive payoffs based on their strategies. In the selection step, the population is updated probabilistically based on the current population and the players' payoffs. Selection processes provide a general framework for the evolutionary games used to model processes based on natural selection \citep{maynardsmith:CUP:1982}.

Roughly speaking, the processes modeled by evolutionary game theory may be split into two classes: cultural and genetic \citep{mcavoy:PLOSCB:2015}. In cultural processes, there is a fixed set of players who repeatedly revise their strategies based on some update rule (such as imitation). In genetic processes, strategies are updated via reproduction and genetic inheritance. Naturally, a process need not be one or the other; there may be both cultural and genetic components in an evolutionary process. Unlike purely cultural processes, those with a genetic component have the property that the players themselves, as well as the size of the population, may actually change via births and deaths, thus it does not make sense to speak of a fixed population of players as in requirement (i) of a stochastic game. However, one can choose an arbitrary enumeration of the players in each step of the process and refer to the player labeled $i$ at time $n$ as ``player $i$." Of course, player $i$ at time $m$ and player $i$ at time $n\neq m$ might be different, but a stochastic selection process is a Markov chain and the transition kernel can be defined for any enumeration of the players at a given time. The implicit property that natural selection does not depend on the enumeration of the players must be stated formally in terms of the update rule. Informally, the update rule for a stochastic selection process must satisfy a symmetry condition that guarantees the dynamics do not depend on these enumerations.

As an example, an evolutionary process based on a game with two-strategies in a well-mixed population may be modeled as a Markov chain whose state space is $\left\{0,1,\dots ,N\right\}$ \citep{nowak:BP:2006}. If $S=\left\{A,B\right\}$ is the strategy set, then the state of the population is determined by the number of $A$-players, which is just an integer $i\in\left\{0,1,\dots ,N\right\}$. Alternatively, one may choose an arbitrary enumeration of the players and represent the state of the population as an element $\left(s_{1},\dots ,s_{N}\right)\in S^{N}$. Since the population is well mixed, evolutionary dynamics depend on only the frequency of each strategy in the population, so any two states $\left(s_{1},\dots ,s_{N}\right)$ and $\left(s_{1}',\dots ,s_{N}'\right)$ consisting of the same number of $A$-players should be indistinguishable. In other words, if $\mathbf{T}$ is the transition matrix for an evolutionary game in this population, then $\mathbf{T}_{\pi s,\tau s'}=\mathbf{T}_{s,s'}$ for each $s,s'\in S^{N}$ and $\pi ,\tau\in\mathfrak{S}_{N}$, where $\mathfrak{S}_{N}$ acts on $S^{N}$ by permuting the coordinates. Thus, the Markov chain is more naturally defined on the \textit{quotient space}
\begin{linenomath}
\begin{align}
\mathcal{S} &:= S^{N} / \mathfrak{S}_{N} ,
\end{align}
\end{linenomath}
which is isomorphic to $\left\{0,1,\dots ,N\right\}$ when $S=\left\{A,B\right\}$. (Recall that an action of a group, $G$, on a set, $X$, gives an equivalence relation, $\sim$, on $X$, which is defined by $x\sim x'$ if and only if there exists $g\in G$ with $x'=gx$. The quotient space, $X/G$, is defined as the set of equivalence classes under $\sim$, and the class containing $x$ is denoted by ``$x\bmod G$.") Of course, one may consider strategy spaces with more than two strategies: Suppose that $S=\left\{A_{1},\dots ,A_{n}\right\}$, and, for each $r=1,\dots ,n$, let $\psi_{r}:S^{N}\rightarrow\left\{0,1,\dots ,N\right\}$ be the map sending a strategy profile, $s\in S^{N}$, to the number of players using strategy $A_{r}$ in $s$. Since $\varphi_{r}\left(\pi s\right) =\varphi_{r}\left(s\right)$ for each $\pi\in\mathfrak{S}_{N}$ and $s\in S^{N}$, the map
\begin{linenomath}
\begin{align}
\Psi &: S^{N} \longrightarrow \left\{ \left(k_{1},\dots ,k_{n}\right)\in\left\{0,1,\dots ,N\right\}^{N}\ :\ k_{1}+\cdots +k_{n}=N \right\} \nonumber \\
&: s \longmapsto \Big( \psi_{1}\left(s\right) , \dots , \psi_{n}\left(s\right) \Big)
\end{align}
\end{linenomath}
descends to an isomorphism
\begin{linenomath}
\begin{align}
\widetilde{\Psi} &: S^{N} / \mathfrak{S}_{N} \longrightarrow \left\{ \left(k_{1},\dots ,k_{n}\right)\in\left\{0,1,\dots ,N\right\}^{N}\ :\ k_{1}+\cdots +k_{n}=N \right\} .
\end{align}
\end{linenomath}
Since an evolutionary update rule may be defined on the space of strategy-frequency profiles,
\begin{linenomath}
\begin{align}
\left\{ \left(k_{1},\dots ,k_{n}\right)\in\left\{0,1,\dots ,N\right\}^{N}\ :\ k_{1}+\cdots +k_{n}=N \right\} ,
\end{align}
\end{linenomath}
we see once again that the Markov chain defined by an evolutionary process in this population naturally has $S^{N}/\mathfrak{S}_{N}$ for a state space. We show here that this phenomenon generalizes to arbitrary types of populations and update rules. In the process of establishing this general construction, we must formally define \textit{population state} (\S\ref{subsubsec:popStates}) and \textit{update rule} (\S\ref{subsubsec:updateRules}).

We first assume that the population size, $N$, is fixed. This assumption allows us to place many of the classical (fixed population size) stochastic evolutionary processes into the context of our framework. After discussing the components of a selection process and giving several examples, we formally define \textit{stochastic selection process} in its full generality (covering populations of variable size) in \S\ref{subsec:variablePop}.

Just like a stochastic game, a stochastic selection process consists of a measurable state space, $\mathcal{P}$, and a strategy space for each player. We assume that
\begin{linenomath}
\begin{align}
S_{1}=S_{2}=\cdots =S_{N}=:S ,
\end{align}
\end{linenomath}
so that $\mathbf{S}=S^{N}$. This assumption that the players all have the same strategy space is not restrictive since the dynamics of the process define the evolution of strategies; one can just enlarge each player's strategy space if necessary and let the dynamics ensure that those strategies that are not available to a given player are never used. Before discussing these dynamics, we first need to explore the state space, $\mathcal{P}$:

\subsection{Population states}\label{subsubsec:popStates}

We seek to appropriate the idea of a state space, $\mathcal{P}$, of a stochastic game in order to introduce population structure into an evolutionary process. In fact, a population's spatial structure (such as a graph) is just one component of this space; mutation rates or phenotypes may also be included in this space. Therefore, rather than declaring $\mathcal{P}$ to be a space of population \textit{structures}, we say that $\mathcal{P}$ is the space of population \textit{states}. A population state indicates properties of the players and relationships between the players. (Note that ``population state" in this context does not include the strategies of the players in the population.) If one enumerates these players differently, then there should be a corresponding ``relabeling" of the population state so that these properties and relationships are preserved. Since changing the enumeration of the population amounts to applying an element of $\mathfrak{S}_{N}$ (the symmetric group on $N$ letters) to $\left\{1,\dots ,N\right\}$, it follows that $\mathcal{P}$ must be equipped with a group action of $\mathfrak{S}_{N}$. Thus, if $s\in S^{N}$ is a strategy profile of the population and $\mathscr{P}\in\mathcal{P}$ is a population state, then the pair $\left(s,\mathscr{P}\right)$ represents the same population of players as $\left(\pi s,\pi\mathscr{P}\right)$ whenever $\pi\in\mathfrak{S}_{N}$. In other words, the population state space, $\mathcal{P}=\mathcal{P}_{N}$, which we write with a subscript to indicate the population size, is a measurable $\mathfrak{S}_{N}$-space. More formally:

\begin{definition}\label{def:structureSpace}
A \textit{population state space} for a set of $N$ players is a measurable space, $\mathcal{P}_{N}$, equipped with an action of $\mathfrak{S}_{N}$ in such a way that the map $\pi :\mathcal{P}_{N}\rightarrow\mathcal{P}_{N}$ is measurable for each $\pi\in\mathfrak{S}_{N}$. If $\mathcal{P}_{N}$ is a population state space for a set of $N$ players, then a \textit{population state} for these players is simply an element $\mathscr{P}\in\mathcal{P}_{N}$.
\end{definition}

\subsubsection{Examples}

\begin{example}[Graphs]
Consider the set of $N\times N$, nonnegative matrices over $\mathbb{R}$,
\begin{linenomath}
\begin{align}
\mathcal{P}_{N}^{\textrm{G}} &:= \left\{ \Gamma\in\mathbb{R}^{N\times N}\ :\ \Gamma_{ij}\geqslant 0\textrm{ for each }i,j=1,\dots ,N \right\} ,
\end{align}
\end{linenomath}
equipped with an action of $\mathfrak{S}_{N}$ defined by $\left(\pi\Gamma\right)_{ij}=\Gamma_{\pi\left(i\right)\pi\left(j\right)}$. An element $\Gamma\in\mathcal{P}_{N}^{\textrm{G}}$ defines a directed, weighted graph whose vertices are $\left\{1,\dots ,N\right\}$, with an edge from $i$ to $j$ if and only if $\Gamma_{ij}\neq 0$ (the weight of the edge is then $\Gamma_{ij}$). $\Gamma$ is \textit{undirected} if $\Gamma_{ij}=\Gamma_{ji}$ for each $i$ and $j$, and $\Gamma$ is \textit{unweighted} if $\Gamma\in\left\{0,1\right\}^{N\times N}$.
\end{example}

\begin{example}[Sets]
A set-structured population consists of a finite number of sets, each containing some subset of the population, such that each player is in at least one set \citep{tarnita:PNAS:2009}. Set-structured populations may be modeled using relations:
\begin{linenomath}
\begin{align}
\mathcal{P}_{N}^{\textrm{S}} &:= \Big\{ R\subseteq\left\{1,\dots ,N\right\}\times\left\{1,\dots ,N\right\}\ :\ R\textrm{ is reflexive and symmetric} \Big\} .
\end{align}
\end{linenomath}
That is, if $R\in\mathcal{P}_{N}^{\textrm{S}}$, then $\left(i,i\right)\in R$ for each $i$ (``reflexive") and $\left(i,j\right)\in R$ if and only if $\left(j,i\right)\in R$ (``symmetric"). $R\in\mathcal{P}_{N}^{\textrm{S}}$ defines a set-structure with $i$ and $j$ in a common set if and only if $\left(i,j\right)\in R$. There is a natural action of $\mathfrak{S}_{N}$ on $\mathcal{P}_{N}^{\textrm{S}}$ defined by
\begin{linenomath}
\begin{align}
\left(i,j\right)\in\pi R\iff\left(\pi i,\pi j\right)\in R,
\end{align}
\end{linenomath}
which makes $\mathcal{P}_{N}^{\textrm{S}}$ into a population state space.
\end{example}

\begin{example}[Demes]
A deme-structure on a population is a subdivision of the population into subpopulations, or ``demes" \citep{taylor:Selection:2000,wakeley:TPB:2004,hauert:JTB:2012}. Similar to set-structured populations, deme-structured populations may be modeled using relations (but with the stronger notion of \textit{equivalence relation}):
\begin{linenomath}
\begin{align}
\mathcal{P}_{N}^{\textrm{D}} &:= \Big\{ R\subseteq\left\{1,\dots ,N\right\}\times\left\{1,\dots ,N\right\}\ :\ R\textrm{ is reflexive, symmetric, and transitive} \Big\} .
\end{align}
\end{linenomath}
$\mathcal{P}_{N}^{\textrm{D}}\subseteq\mathcal{P}_{N}^{\textrm{S}}$, and whereas a set-structured population may have overlapping sets, a deme-structured population has disjoint sets. The additional transitivity requirement guarantees that these sets partition $\left\{1,\dots ,N\right\}$. The action of $\mathfrak{S}_{N}$ on $\mathcal{P}_{N}^{\textrm{D}}$ is the one inherited from $\mathcal{P}_{N}^{\textrm{S}}$, making $\mathcal{P}_{N}^{\textrm{D}}$ into a population state space.
\end{example}

So far, we have considered population structures that describe spatial and qualitative relationships between the players. One could also associate to the players quantities such as \textit{mutation rates} or \textit{phenotypes}:

\begin{example}[Mutation rates]
Consider a process in which updates are based on births and deaths (such as a Moran or Wright-Fisher process). Moreover, suppose that the spatial structure of the population is a graph. If player $i$ reproduces, then with probability $\varepsilon_{i}$ the offspring adopts a novel strategy uniformly at random (``mutates"), and with probability $1-\varepsilon_{i}$ the offspring inherits the strategy of the parent. The mutation rate, $\varepsilon_{i}$, is passed on directly from parent to offspring. The population state space for this process is then $\mathcal{P}_{N}:=\mathcal{P}_{N}^{\textrm{G}}\times\left[0,1\right]^{N}$; a population state consists of (i) a graph, indicating the spatial relationships between the players, and (ii) a profile of mutation rates, $\varepsilon\in\left[0,1\right]^{N}$, with $\varepsilon_{i}$ indicating the probability that the offspring of player $i$ mutates. For update rules based on imitation, mutation rates might more appropriately be called ``exploration rates," and they are implemented slightly differently. In general, mutation rates appear in different forms and help to distinguish cultural and genetic update rules; we give several examples in \S\ref{subsubsec:updateRules}. The upshot of this discussion is that a population state may consist of a spatial structure, such as a graph in $\mathcal{P}_{N}^{\textrm{G}}$, \textit{as well} as some extra information pertaining to the players, such as mutation rates.
\end{example}

\begin{example}[Phenotype space]
In addition to strategies, the players may also have \textit{phenotypes} that affect interactions with other players in the population \citep{antal:PNAS:2009,nowak:PTRSB:2009}. If the spatial structure of the population is a graph and the phenotype of each player is a one-dimensional discrete quantity, then one may define the population state space to be $\mathcal{P}_{N}^{\textrm{G}}\times\mathbb{Z}^{N}$. Just as with mutation rates, the rest of the process determines how these phenotypes affect the dynamics. In \S\ref{subsubsec:aggregatePayoffFunctions}, we continue this example and go into the details of how the inclusion of phenotypes in the population state space affects the payoffs of the players, which can then be used to recover a model of evolution in \textit{phenotype space} of \citet{antal:PNAS:2009}.
\end{example}

\subsubsection{Symmetries of population states}

The action of $\mathfrak{S}_{N}$ on $\mathcal{P}_{N}$ can be used to formally define a notion of population symmetry:
\begin{definition}[Automorphism of population state]
For a population state, $\mathscr{P}$, in a population state space, $\mathcal{P}_{N}$, an \textit{automorphism} of $\mathscr{P}$ is an element $\pi\in\mathfrak{S}_{N}$ such that $\pi\mathscr{P}=\mathscr{P}$. The group of automorphisms of $\mathscr{P}$ is $\textrm{Aut}\left(\mathscr{P}\right) :=\textrm{Stab}_{\mathfrak{S}_{N}}\left(\mathscr{P}\right)$, where $\textrm{Stab}_{\mathfrak{S}_{N}}\left(\mathscr{P}\right)$ denotes the stabilizer of $\mathscr{P}$ under the group action of $\mathfrak{S}_{N}$ on $\mathcal{P}_{N}$.
\end{definition}

If $\Gamma\in\mathcal{P}_{N}^{\textrm{G}}$, for example, then $\textrm{Aut}\left(\Gamma\right)$ is the group of graph automorphisms of $\Gamma$ in the classical sense, i.e. the set of $\pi\in\mathfrak{S}_{N}$ such that $\Gamma_{\pi\left(i\right)\pi\left(j\right)}=\Gamma_{ij}$ for each $i$ and $j$. Such automorphisms have played an important role in the study of evolutionary games on graphs \citep[see][]{taylor:Nature:2007,debarre:NC:2014}. We discuss symmetries of population states further in \S\ref{sec:discussion}.

Now that we have a formal definition of \textit{population state space}, we explore how this space influences evolutionary dynamics. The processes we seek to model have two timescales: interactions and updates. In \S\ref{subsubsec:aggregatePayoffFunctions}, we consider the influence of the population state on interactions, and in \S\ref{subsubsec:updateRules}, we define \textit{update rule} and show how the population state fits into the update step of an evolutionary process.

\subsection{Aggregate payoff functions}\label{subsubsec:aggregatePayoffFunctions}

Prior to stating the definition of a general payoff function for a stochastic selection process, we consider a motivating example that is based on a popular type of game used to model frequency-dependent fitness in evolutionary game theory:
\begin{example}
Consider the symmetric, two-player game whose payoff matrix is
\begin{linenomath}
\begin{align}\label{generic2by2}
\bordermatrix{%
 & A_{1} & A_{2} \cr
A_{1} &\ a_{11} & \ a_{12} \cr
A_{2} &\ a_{21} & \ a_{22} \cr
} .
\end{align}
\end{linenomath}
If the population structure is a graph and the population state space is $\mathcal{P}_{N}^{\textrm{G}}$, then one can construct a function $u:\left\{1,2\right\}^{N}\times\mathcal{P}_{N}^{\textrm{G}}\rightarrow\mathbb{R}^{N}$ by letting $u_{i}$ be defined as
\begin{linenomath}
\begin{align}\label{payoff2by2}
u_{i}\Big(\left(s_{1},\dots ,s_{N}\right) ,\Gamma\Big) &:= \sum_{j=1}^{N}\Gamma_{ij}a_{s_{i}s_{j}} .
\end{align}
\end{linenomath}
That is, $u_{i}$ is the ``aggregate payoff" function for player $i$ since it produces the total payoff from player $i$'s interactions with all of his or her neighbors, weighted appropriately by the edge weights of the population structure. If $\pi\in\mathfrak{S}_{N}$, then
\begin{linenomath}
\begin{align}\label{symmetryCalculation}
u_{i}\Big(\left(s_{\pi\left(1\right)},\dots ,s_{\pi\left(N\right)}\right) ,\pi\Gamma\Big) &= \sum_{j=1}^{N}\Gamma_{\pi\left(i\right)\pi\left(j\right)}a_{s_{\pi\left(i\right)}s_{\pi\left(j\right)}} \nonumber \\
&= \sum_{j=1}^{N}\Gamma_{\pi\left(i\right) j}a_{s_{\pi\left(i\right)}s_{j}} \nonumber \\
&= u_{\pi\left(i\right)}\Big(\left(s_{1},\dots ,s_{N}\right) ,\Gamma\Big) .
\end{align}
\end{linenomath}
Therefore, although $u^{\Gamma}:=u\left(\-- ,\Gamma\right) :\left\{1,2\right\}^{N}\rightarrow\mathbb{R}^{N}$ need not be symmetric in the sense that $u^{\Gamma}\left(\pi s\right) =\pi u^{\Gamma}\left(s\right)$ for each $s\in\left\{1,2\right\}^{N}$, it \textit{is} symmetric in the sense that
\begin{linenomath}
\begin{align}\label{eq:symmetricPayoff}
u\left( \pi s , \pi\Gamma \right) &= \pi u\left(s,\Gamma\right)
\end{align}
\end{linenomath}
for each $s\in\left\{1,2\right\}^{N}$ and $\Gamma\in\mathcal{P}_{N}=\mathcal{P}_{N}^{\textrm{G}}$. In other words, all of the information that results in payoff asymmetry is contained in the population state space, $\mathcal{P}_{N}$.
\end{example}

Using the function $u$ and Eq. (\ref{eq:symmetricPayoff}) of the previous example as motivation, we have:
\begin{definition}[Aggregate payoff function]\label{def:aggregatePayoff}
An \textit{aggregate payoff function} is a map
\begin{linenomath}
\begin{align}
u &: S^{N}\times\mathcal{P}_{N} \longrightarrow \mathbb{R}^{N}
\end{align}
\end{linenomath}
that satisfies $u\left(\pi s,\pi\mathscr{P}\right) =\pi u\left(s,\mathscr{P}\right)$ for each $\pi\in\mathfrak{S}_{N}$, $s\in S^{N}$, and $\mathscr{P}\in\mathcal{P}_{N}$.
\end{definition}

The symmetry condition in Definition \ref{def:aggregatePayoff}, $u\left(\pi s,\pi\mathscr{P}\right) =\pi u\left(s,\mathscr{P}\right)$, implies that an aggregate payoff function, $u$, is completely determined by the map $u_{1}:S^{N}\times\mathcal{P}_{N}\rightarrow\mathbb{R}$. Indeed, if $u_{1}$ is known and $\pi\in\mathfrak{S}_{N}$ sends $1$ to $i$, then $u_{i}\left(s,\mathscr{P}\right) =u_{1}\left(\pi s,\pi\mathscr{P}\right)$ for each $s\in S^{N}$ and $\mathscr{P}\in\mathcal{P}_{N}$, which recovers the map $u:S^{N}\times\mathcal{P}_{N}\rightarrow\mathbb{R}^{N}$. This symmetry condition must hold even for individual encounters that are \textit{asymmetric} (Example \ref{ex:asymmetric2by2}).

\subsubsection{Examples}

\begin{example}\label{ex:asymmetric2by2}
In place of (\ref{generic2by2}), one could consider a collection of \textit{bimatrices},
\begin{linenomath}
\begin{align}\label{eq:asymmetric2by2}
\mathbf{M}^{ij} &:= \bordermatrix{%
 & A_{1} & A_{2} \cr
A_{1} &\ a_{11}^{ij},\ a_{11}^{ji} & \ a_{12}^{ij},\ a_{21}^{ji} \cr
A_{2} &\ a_{21}^{ij},\ a_{12}^{ji} & \ a_{22}^{ij},\ a_{22}^{ji} \cr
} ,
\end{align}
\end{linenomath}
indexed by $i,j\in\left\{1,\dots ,N\right\}$ \citep{mcavoy:PLOSCB:2015}. For each $i$ and $j$, $\mathbf{M}^{ij}$ is the payoff matrix for player $i$ against player $j$, with the first coordinate of each entry denoting the payoff to player $i$ and the second coordinate denoting the payoff to player $j$. The collection $\left\{\mathbf{M}^{ij}\right\}_{i,j=1}^{N}$ is equivalent to an element of $\left(\mathbb{R}^{2\times 2}\right)^{N\times N}$, i.e. a $2\times 2$ real matrix indicating the payoff to player $i$ against player $j$ for each $i,j=1,\dots ,N$. In this case, the population state space consists of more than spatial structures; it also includes the details of the payoff-asymmetry appearing in individual encounters. In other words, if the population is graph-structured, then a population state consists of a graph and a collection of payoff matrices of the form of Eq. (\ref{eq:asymmetric2by2}), i.e.
\begin{linenomath}
\begin{align}\label{enlargedInteractionSpace}
\mathcal{P}_{N} &:= \mathcal{P}_{N}^{\textrm{G}}\times\left(\mathbb{R}^{2\times 2}\right)^{N\times N} ,
\end{align}
\end{linenomath}
where the action of $\pi\in\mathfrak{S}_{N}$ on $\left(\mathbb{R}^{2\times 2}\right)^{N\times N}$ is $\pi\left( \mathbf{M}^{ij} \right)_{i,j=1}^{N}:=\left( \mathbf{M}^{\pi\left(i\right)\pi\left(j\right)} \right)_{i,j=1}^{N}$. The aggregate payoff function, $u:S^{N}\times\mathcal{P}_{N}\rightarrow\mathbb{R}^{N}$, is defined by
\begin{linenomath}
\begin{align}\label{asymmetricPayoff}
u_{i}\Big(\left(s_{1},\dots ,s_{N}\right) ,\left(\Gamma ,\left(\mathbf{M}^{ij}\right)_{i,j=1}^{N}\right)\Big) &:= \sum_{j=1}^{N}\Gamma_{ij}a_{s_{i}s_{j}}^{ij} .
\end{align}
\end{linenomath}
For $\pi\in\mathfrak{S}_{N}$, we see that
\begin{linenomath}
\begin{align}
u_{i}\Big(\left(s_{\pi\left(1\right)},\dots ,s_{\pi\left(N\right)}\right) ,\pi\left(\Gamma ,\left(\mathbf{M}^{ij}\right)_{i,j=1}^{N}\right)\Big) &= \sum_{j=1}^{N}\Gamma_{\pi\left(i\right)\pi\left(j\right)}a_{s_{\pi\left(i\right)}s_{\pi\left(j\right)}}^{\pi\left(i\right)\pi\left(j\right)} \nonumber \\
&= \sum_{j=1}^{N}\Gamma_{\pi\left(i\right) j}a_{s_{\pi\left(i\right)}s_{j}}^{\pi\left(i\right) j} \nonumber \\
&= u_{\pi\left(i\right)}\Big(\left(s_{1},\dots ,s_{N}\right) ,\left(\Gamma ,\left(\mathbf{M}^{ij}\right)_{i,j=1}^{N}\right)\Big) ,
\end{align}
\end{linenomath}
so $u$ defines an aggregate payoff function in the sense of Definition \ref{def:aggregatePayoff}.
\end{example}

\begin{remark}
For a fixed collection, $\left(\mathbf{M}^{ij}\right)_{i,j=1}^{N}$, we could have instead let $\mathcal{P}_{N}=\mathcal{P}_{N}^{\textrm{G}}$ and
\begin{linenomath}
\begin{align}\label{asymmetricPayoffWrongWay}
u_{i}\Big(\left(s_{1},\dots ,s_{N}\right) ,\Gamma\Big) &:= \sum_{j=1}^{N}\Gamma_{ij}a_{s_{i}s_{j}}^{ij} .
\end{align}
\end{linenomath}
However, for $\pi\in\mathfrak{S}_{N}$, we would then have
\begin{linenomath}
\begin{subequations}
\begin{align}
u_{i}\Big(\left(s_{\pi\left(1\right)},\dots ,s_{\pi\left(N\right)}\right) ,\pi\Gamma\Big) &= \sum_{j=1}^{N}\Gamma_{\pi\left(i\right) j}a_{s_{\pi\left(i\right)}s_{j}}^{i\pi^{-1}\left(j\right)} ; \label{firstAsymmetricFunction} \\
u_{\pi\left(i\right)}\Big(\left(s_{1},\dots ,s_{N}\right) ,\Gamma\Big) &= \sum_{j=1}^{N}\Gamma_{\pi\left(i\right) j}a_{s_{\pi\left(i\right)}s_{j}}^{\pi\left(i\right) j} .\label{secondAsymmetricFunction}
\end{align}
\end{subequations}
\end{linenomath}
The only way (\ref{firstAsymmetricFunction}) and (\ref{secondAsymmetricFunction}) are the same for each $s\in\left\{1,2\right\}^{N}$ and $\Gamma\in\mathcal{P}_{N}^{\textrm{G}}$ is if $\mathbf{M}^{ij}=\mathbf{M}^{\pi\left(i\right)\pi\left(j\right)}$ for each $i,j=1,\dots ,N$, and this equality need not hold (which would mean that $u$, when defined in this way, is not an aggregate payoff function in the sense of Definition \ref{def:aggregatePayoff}). Therefore, by enlarging $\mathcal{P}_{N}$ via (\ref{enlargedInteractionSpace}) and defining $u$ via (\ref{asymmetricPayoff}), we can essentially ``factor out" the asymmetry present in the payoff function defined by (\ref{asymmetricPayoffWrongWay}). In other words, $\mathcal{P}_{N}$ contains all of the non-strategy information that distinguishes the players' payoffs.
\end{remark}

Due to the separation of timescales in the selection processes we consider here, it often happens that $u$ is independent of a portion of the population state space. More specifically, the population state space can be decomposed into an \textit{interaction state space}, $\mathcal{E}_{N}$, and a \textit{dispersal state space}, $\mathcal{D}_{N}$, such that $\mathcal{P}_{N}=\mathcal{E}_{N}\times\mathcal{D}_{N}$ and
\begin{linenomath}
\begin{align}
u\Big(s,\left(\mathscr{E},\mathscr{D}\right)\Big) &= u\Big(s,\left(\mathscr{E},\mathscr{D}'\right)\Big)
\end{align}
\end{linenomath}
for each $s\in S^{N}$, $\mathscr{E}\in\mathcal{E}_{N}$, and $\mathscr{D},\mathscr{D}'\in\mathcal{D}_{N}$. This decomposition generalizes models with separate interaction and dispersal \textit{graphs} \citep{taylor:Nature:2007,ohtsuki:PRL:2007,ohtsuki:JTB:2007,pacheco:PLoSCB:2009,debarre:NC:2014}. The interaction state, for instance, might consist of a population structure \textit{and} other information (such as phenotypes):

\begin{example}[Phenotype space, continued]
\citet{antal:PNAS:2009} study the evolution of cooperation in phenotype space. In terms of (\ref{generic2by2}), each player has a strategy, $A_{1}$ (``cooperate") or $A_{2}$ (``defect"), as well as a one-dimensional phenotype, which is simply an integer. If an interaction state has a graph as its spatial structure, then the interaction state space is $\mathcal{E}_{N}:=\mathcal{P}_{N}^{\textrm{G}}\times\mathbb{Z}^{N}$. Thus, an interaction structure consists of a graph, $\Gamma\in\mathcal{P}_{N}^{\textrm{G}}$, and an $N$-tuple of phenotypes $r=\left(r_{1},\dots ,r_{N}\right)\in\mathbb{Z}^{N}$, where $r_{i}$ is the phenotype of player $i$ in the population. The phenotypes affect the strategies of the players as follows: cooperators cooperate with other neighbors with whom they share a phenotype, and they defect otherwise. Defectors always defect, regardless of phenotypic similarities. For each $i$, the payoff function, $u:\left\{1,2\right\}^{N}\times\mathcal{E}_{N}\rightarrow\mathbb{R}^{N}$, satisfies
\begin{linenomath}
\begin{align}
u_{i}\Big(\left(s_{1},\dots ,s_{N}\right) ,\left(\Gamma ,r\right)\Big) &= \sum_{j=1}^{N}\Gamma_{ij}\Big( \delta_{r_{i},r_{j}}a_{s_{i}s_{j}} + \left(1-\delta_{r_{i},r_{j}}\right) a_{22} \Big) ,
\end{align}
\end{linenomath}
where $\delta_{r_{i},r_{j}}$ is $1$ if $r_{i}=r_{j}$ and $0$ otherwise. Therefore, one can directly implement the influence of phenotype on strategy using the interaction state space, $\mathcal{E}_{N}$, and $u$.
\end{example}

\subsection{Update rules}\label{subsubsec:updateRules}

We saw at the beginning of \S\ref{sec:ssp} that for a game with $n$ strategies in a well-mixed population, the state space for an evolutionary process is
\begin{linenomath}
\begin{align}\label{eq:stateSpaceWellMixed}
S^{N} / \mathfrak{S}_{N} &\cong \left\{ \left(k_{1},\dots ,k_{n}\right)\in\left\{0,1,\dots ,N\right\}^{N}\ :\ k_{1}+\cdots +k_{n}=N \right\} .
\end{align}
\end{linenomath}
On the other hand, between update steps, one can simply fix some enumeration of the population and represent the state of the evolutionary process by an element of $S^{N}$, i.e. a \textit{representative} of the state space, $S^{N}/\mathfrak{S}_{N}$. In populations with spatial structure, mutation rates, phenotypic differences, etc., this representative contains more information than simply a strategy profile; it also contains information about the population state. In other words, at a fixed point in time, the state of the evolutionary process can be described by an element of $S^{N}\times\mathcal{P}_{N}$, where $\mathcal{P}_{N}$ is a population state space. Of course, the evolutionary dynamics of the process should not be affected by the choice of enumeration of the players in each time step, which means that the state space for the evolutionary process is naturally the quotient space
\begin{linenomath}
\begin{align}
\mathcal{S} &:= \left(S^{N}\times\mathcal{P}_{N}\right) /\mathfrak{S}_{N} ,
\end{align}
\end{linenomath}
generalizing the state space of Eq. (\ref{eq:stateSpaceWellMixed}) to structured populations.

We now wish to describe the update step of an evolutionary process on $\mathcal{S}$. This update rule should not depend on how the players in the updated population are labeled. For example, if a player dies and is replaced by the offspring of another player, then the result of this death and replacement is a new element of $\mathcal{S}$. In other words, the new population does \textit{not} lie in $S^{N}\times\mathcal{P}_{N}$ in a natural way; we must choose an enumeration of the players in order to get an element of $S^{N}\times\mathcal{P}_{N}$, and this enumeration may be arbitrary. Therefore, given the current strategy profile and population state, an update rule should give a probability distribution over the state space of the process, $\mathcal{S}$ (\textit{not} $S^{N}\times\mathcal{P}_{N}$). On the other hand, in order to update the state of the population, one needs to speak of the likelihood that each player in the current state is updated. To do so, one may choose a representative of the current state of the process, $\left(s,\mathscr{P}\right)\in S^{N}\times\mathcal{P}_{N}$, which is equivalent to choosing a labeling of the players at that point in time. Again, the distribution over $\mathcal{S}$ (conditioned on the current state of the process) should not depend on the labeling of the current state. Finally, this distribution over $\mathcal{S}$ is a function of the \textit{fitness profile} of the population; each player has a real-valued fitness, and the update rule depends on these values. Putting these components together, we have:

\begin{definition}[Update rule]\label{def:updateRule}
An \textit{update rule} is a map,
\begin{linenomath}
\begin{align}
\mathcal{U} &: \mathbb{R}^{N} \longrightarrow K\left(S^{N}\times\mathcal{P}_{N},\mathcal{S}\right) ,
\end{align}
\end{linenomath}
that satisfies the symmetry condition
\begin{linenomath}
\begin{align}\label{eq:symmetricUpdate}
\mathcal{U}\left[\pi x\right]\Big(\left(\pi s,\pi\mathscr{P}\right) ,E\Big) &= \mathcal{U}\left[x\right]\Big(\left(s,\mathscr{P}\right) ,E\Big)
\end{align}
\end{linenomath}
for each $\pi\in\mathfrak{S}_{N}$, $x\in\mathbb{R}^{N}$, $\left(s,\mathscr{P}\right)\in S^{N}\times\mathcal{P}_{N}$, and $E\in\mathcal{F}\left(\mathcal{S}\right)$, where $\mathcal{F}\left(\mathcal{S}\right)$ is the quotient $\sigma$-algebra on $\mathcal{S}=\left(S^{N}\times\mathcal{P}_{N}\right) /\mathfrak{S}_{N}$ derived from $\mathcal{F}\left(S\right)$ and $\mathcal{F}\left(\mathcal{P}_{N}\right)$.
\end{definition}

That is, an update rule is a family of Markov kernels,
\begin{linenomath}
\begin{align}
\Big\{\mathcal{U}\left[x\right]\Big\}_{x\in\mathbb{R}^{N}}\subseteq K\left(S^{N}\times\mathcal{P}_{N},\mathcal{S}\right) ,
\end{align}
\end{linenomath}
parametrized by the fitness profiles of the population, $x\in\mathbb{R}^{N}$, and satisfying Eq. (\ref{eq:symmetricUpdate}). Eq. (\ref{eq:symmetricUpdate}) says that the update does not depend on how the current population is represented. In other words, if $\left(s,\mathscr{P}\right)$ and $\left(\pi s,\pi\mathscr{P}\right)$ are two representatives of the same population at time $t$, then the update rule treats $\left(s,\mathscr{P}\right)$ and $\left(\pi s,\pi\mathscr{P}\right)$ as the same population. (If $x$ is the fitness profile corresponding to the representative $\left(s,\mathscr{P}\right)$, then $\pi x$ is the fitness profile corresponding to the representative $\left(\pi s,\pi\mathscr{P}\right)$.)

Together, an update rule and aggregate payoff function define a Markov chain on $\mathcal{S}$ whose kernel, $\kappa$, is constructed as follows: Let $f:\mathbb{R}\rightarrow\mathbb{R}$ be a \textit{payoff-to-fitness} map, i.e. a function that converts a player's payoff to fitness. Consider the function
\begin{linenomath}
\begin{align}
F &: \mathbb{R}^{N} \longrightarrow \mathbb{R}^{N} \nonumber \\
&: \left(x_{1},\dots ,x_{N}\right) \longmapsto \Big( f\left(x_{1}\right) , \dots ,f\left(x_{N}\right) \Big) ,
\end{align}
\end{linenomath}
which converts payoff profiles to fitness profiles. If $u:S^{N}\times\mathcal{P}_{N}\rightarrow\mathbb{R}^{N}$ is an aggregate payoff function, then, for $\left(s,\mathscr{P}\right)\in S^{N}\times\mathcal{P}_{N}$ and $E\in\mathcal{F}\left(\mathcal{S}\right)$, we let
\begin{linenomath}
\begin{align}\label{eq:updateKernel}
\kappa\Big( \left(s,\mathscr{P}\right)\bmod\mathfrak{S}_{N} , E \Big) &:= \mathcal{U}\left[F\Big(u\left(s,\mathscr{P}\right)\Big)\right]\Big( \left(s,\mathscr{P}\right) , E \Big) .
\end{align}
\end{linenomath}
$\kappa$ is well defined since, for each $\pi\in\mathfrak{S}_{N}$,
\begin{linenomath}
\begin{align}
\kappa\Big( \left(\pi s,\pi\mathscr{P}\right)\bmod\mathfrak{S}_{N} , E \Big) &= \mathcal{U}\left[F\Big(u\left(\pi s,\pi\mathscr{P}\right)\Big)\right]\Big( \left(\pi s,\pi\mathscr{P}\right) , E \Big) \nonumber \\
&= \mathcal{U}\left[\pi F\Big(u\left(s,\mathscr{P}\right)\Big)\right]\Big( \left(\pi s,\pi\mathscr{P}\right) , E \Big) \nonumber \\
&= \mathcal{U}\left[F\Big(u\left(s,\mathscr{P}\right)\Big)\right]\Big( \left(s,\mathscr{P}\right) , E \Big) \nonumber \\
&= \kappa\Big( \left(s,\mathscr{P}\right)\bmod\mathfrak{S}_{N} , E \Big) , \label{eq:kappaWellDefined}
\end{align}
\end{linenomath}
where the second and third lines come from Eqs. (\ref{eq:symmetricPayoff}) and (\ref{eq:symmetricUpdate}), respectively.

\subsubsection{Update pre-rules}\label{subsubsec:updatePreRules}

Despite the fact that the evolutionary processes we seek to model here naturally have $\mathcal{S}=\left(S^{N}\times\mathcal{P}_{N}\right) /\mathfrak{S}_{N}$ for a state space, many evolutionary processes in the literature are defined directly on $S^{N}\times\mathcal{P}_{N}$ \citep[see][]{allen:JMB:2012}. Update rules are sometimes cumbersome to write out explicitly, and defining a Markov chain on $S^{N}\times\mathcal{P}_{N}$ instead of on $\left(S^{N}\times\mathcal{P}_{N}\right) /\mathfrak{S}_{N}$ can simplify the presentation of the transition kernel. In this context, the notion of ``update rule" still makes sense, but we instead call it an \textit{update pre-rule} to distinguish it from the update rule of Definition \ref{def:updateRule}:

\begin{definition}[Update pre-rule]\label{def:updatePreRule}
An \textit{update pre-rule} is a map,
\begin{linenomath}
\begin{align}
\mathcal{U}_{0} &: \mathbb{R}^{N} \longrightarrow K\left(S^{N}\times\mathcal{P}_{N},S^{N}\times\mathcal{P}_{N}\right) ,
\end{align}
\end{linenomath}
such that for each $\pi\in\mathfrak{S}_{N}$, $x\in\mathbb{R}^{N}$, $\left(s,\mathscr{P}\right)\in S^{N}\times\mathcal{P}_{N}$, and $E\in\mathcal{F}\left(\mathcal{S}\right)$,
\begin{linenomath}
\begin{align}\label{eq:updatePreRuleSymmetry}
\mathcal{U}_{0}\left[\pi x\right]\Big( \left(\pi s,\pi\mathscr{P}\right) , E \Big) &= \mathcal{U}_{0}\left[x\right]\Big( \left(s,\mathscr{P}\right) , \tau E \Big)
\end{align}
\end{linenomath}
for some $\tau\in\mathfrak{S}_{N}$.
\end{definition}

In many cases, the permutation $\tau$ is just $\pi^{-1}$, i.e.
\begin{linenomath}
\begin{align}
\mathcal{U}_{0}\left[\pi x\right]\Big( \left(\pi s,\pi\mathscr{P}\right) , \pi E \Big) &= \mathcal{U}_{0}\left[x\right]\Big( \left(s,\mathscr{P}\right) , E \Big)
\end{align}
\end{linenomath}
for each $\pi\in\mathfrak{S}_{N}$. However, all that an evolutionary process on $S^{N}\times\mathcal{P}_{N}$ really requires is that if state $\left(\pi s,\pi\mathscr{P}\right)$ is updated to $\left(s',\mathscr{P}\right)$, then state $\left(s,\mathscr{P}\right)$ is updated to $\left(\tau s',\tau\mathscr{P}'\right)$ for some $\tau$. To relate Definitions \ref{def:updateRule} and \ref{def:updatePreRule}, consider the projection map,
\begin{linenomath}
\begin{align}
\Pi &: S^{N}\times\mathcal{P}_{N} \longrightarrow \left(S^{N}\times\mathcal{P}_{N}\right) /\mathfrak{S}_{N} = \mathcal{S} \nonumber \\
&: \left(s,\mathscr{P}\right) \longmapsto \left(s,\mathscr{P}\right)\bmod\mathfrak{S}_{N} .
\end{align}
\end{linenomath}
$\Pi$ gives rise to a pushforward map on measures,
\begin{linenomath}
\begin{align}
\Pi_{\ast} &: \Delta\left(S^{N}\times\mathcal{P}_{N}\right) \longrightarrow \Delta\left(\mathcal{S}\right) \nonumber \\
&: \mu \longmapsto \mu\circ\Pi^{-1} , \label{eq:pushforwardOnMeasures}
\end{align}
\end{linenomath}
which can be used to naturally derive an update rule from an update pre-rule:

\begin{proposition}\label{prop:preRuleToRule}
An update pre-rule canonically defines an update rule.
\end{proposition}
\begin{proof}
Let $\mathcal{U}_{0}$ be an update pre-rule and consider the map
\begin{linenomath}
\begin{align}
\Pi_{\ast}\mathcal{U}_{0} &: \mathbb{R}^{N} \longrightarrow K\left(S^{N}\times\mathcal{P}_{N},\mathcal{S}\right) \nonumber \\
&: x \longmapsto \Big\{ \left(s,\mathscr{P}\right) \mapsto \Pi_{\ast}\mathcal{U}_{0}\left[x\right]\Big(\left(s,\mathscr{P}\right) ,\--\Big) \Big\} . \label{eq:preRuleToRule}
\end{align}
\end{linenomath}
For $\pi\in\mathfrak{S}_{N}$, $x\in\mathbb{R}$, $\left(s,\mathscr{P}\right)\in S^{N}\times\mathcal{P}_{N}$, and $E\in\mathcal{F}\left(\mathcal{S}\right)$, there exists $\tau\in\mathfrak{S}_{N}$ such that
\begin{linenomath}
\begin{align}
\left(\Pi_{\ast}\mathcal{U}_{0}\right)\left[\pi x\right]\Big(\left(\pi s,\pi\mathscr{P}\right) ,E\Big) &= \mathcal{U}_{0}\left[\pi x\right]\Big(\left(\pi s,\pi\mathscr{P}\right) ,\Pi^{-1}E\Big) \nonumber \\
&= \mathcal{U}_{0}\left[x\right]\Big(\left(s,\mathscr{P}\right) ,\tau\Pi^{-1}E\Big) \nonumber \\
&= \mathcal{U}_{0}\left[x\right]\Big(\left(s,\mathscr{P}\right) ,\Pi^{-1}E\Big) \nonumber \\
&= \left(\Pi_{\ast}\mathcal{U}_{0}\right)\left[x\right]\Big(\left(s,\mathscr{P}\right) ,E\Big) ,
\end{align}
\end{linenomath}
so $\Pi_{\ast}\mathcal{U}_{0}$ is an update rule, which completes the proof.
\end{proof}

In other words, an update pre-rule can be ``pushed forward" to an update rule. If $S$ and $\mathcal{P}_{N}$ are finite, then an update rule can also be ``pulled back" to an update pre-rule:
\begin{proposition}\label{prop:updatePullback}
If $\mathcal{U}$ is an update rule and $S$ and $\mathcal{P}_{N}$ are finite, then there exists an update pre-rule, $\mathcal{U}_{0}$, such that $\Pi_{\ast}\mathcal{U}_{0}=\mathcal{U}$. That is, $\mathcal{U}$ can be ``pulled back" to $\mathcal{U}_{0}$.
\end{proposition}
\begin{proof}
From $\mathcal{U}$, we define a map, $\mathcal{U}_{0}$, as follows:
\begin{linenomath}
\begin{align}
\mathcal{U}_{0}\left[x\right]\Big( \left(s,\mathscr{P}\right) ,\left(s',\mathscr{P}'\right) \Big) &= \frac{1}{\left|\textrm{orb}_{\mathfrak{S}_{N}}\left(s',\mathscr{P}'\right)\right|}\mathcal{U}\left[x\right]\Big( \left(s,\mathscr{P}\right) , \left(s',\mathscr{P}'\right)\bmod\mathfrak{S}_{N} \Big) .
\end{align}
\end{linenomath}
For $\pi\in\mathfrak{S}_{N}$, we see from Eq. (\ref{eq:symmetricUpdate}) that
\begin{linenomath}
\begin{align}
\mathcal{U}_{0} &\left[\pi x\right]\Big( \left(\pi s,\pi\mathscr{P}\right) ,\left(s',\mathscr{P}'\right) \Big) \nonumber \\
&= \frac{1}{\left|\textrm{orb}_{\mathfrak{S}_{N}}\left(s',\mathscr{P}'\right)\right|}\mathcal{U}\left[\pi x\right]\Big( \left(\pi s,\pi\mathscr{P}\right) , \left(s',\mathscr{P}'\right)\bmod\mathfrak{S}_{N} \Big) \nonumber \\
&= \frac{1}{\left|\textrm{orb}_{\mathfrak{S}_{N}}\left(s',\mathscr{P}'\right)\right|}\mathcal{U}\left[x\right]\Big( \left(s,\mathscr{P}\right) , \left(s',\mathscr{P}'\right)\bmod\mathfrak{S}_{N} \Big) \nonumber \\
&= \mathcal{U}_{0}\left[x\right]\Big( \left(s,\mathscr{P}\right) ,\left(s',\mathscr{P}'\right) \Big) .
\end{align}
\end{linenomath}
Therefore, $\mathcal{U}_{0}$ satisfies Eq. (\ref{eq:updatePreRuleSymmetry}) and defines an update pre-rule. Since
\begin{linenomath}
\begin{align}
\left(\Pi_{\ast}\mathcal{U}_{0}\right) &\left[x\right]\Big( \left(s,\mathscr{P}\right) , \left(s',\mathscr{P}'\right)\bmod\mathfrak{S}_{N} \Big) \nonumber \\
&= \mathcal{U}_{0}\left[x\right]\left( \left(s,\mathscr{P}\right) , \Pi^{-1}\Big(\left(s',\mathscr{P}'\right)\bmod\mathfrak{S}_{N}\Big) \right) \nonumber \\
&= \sum_{\left(s'',\mathscr{P}''\right)\in\textrm{orb}_{\mathfrak{S}_{N}}\left(s',\mathscr{P}'\right)}\mathcal{U}_{0}\left[x\right]\Big( \left(s,\mathscr{P}\right) ,\left(s'',\mathscr{P}''\right) \Big) \nonumber \\
&= \sum_{\left(s'',\mathscr{P}''\right)\in\textrm{orb}_{\mathfrak{S}_{N}}\left(s',\mathscr{P}'\right)}\frac{1}{\left|\textrm{orb}_{\mathfrak{S}_{N}}\left(s'',\mathscr{P}''\right)\right|}\mathcal{U}\left[x\right]\Big( \left(s,\mathscr{P}\right) , \left(s'',\mathscr{P}''\right)\bmod\mathfrak{S}_{N} \Big) \nonumber \\
&= \sum_{\left(s'',\mathscr{P}''\right)\in\textrm{orb}_{\mathfrak{S}_{N}}\left(s',\mathscr{P}'\right)}\frac{1}{\left|\textrm{orb}_{\mathfrak{S}_{N}}\left(s',\mathscr{P}'\right)\right|}\mathcal{U}\left[x\right]\Big( \left(s,\mathscr{P}\right) , \left(s',\mathscr{P}'\right)\bmod\mathfrak{S}_{N} \Big) \nonumber \\
&= \mathcal{U}\left[x\right]\Big( \left(s,\mathscr{P}\right) , \left(s',\mathscr{P}'\right)\bmod\mathfrak{S}_{N} \Big) ,
\end{align}
\end{linenomath}
it follows that $\Pi_{\ast}\mathcal{U}_{0}=\mathcal{U}$, which completes the proof.
\end{proof}

\begin{remark}
The proof of Proposition \ref{prop:updatePullback} requires that
\begin{linenomath}
\begin{align}
E_{0}\in\mathcal{F}\left(S^{N}\times\mathcal{P}_{N}\right) &\implies \Pi E_{0}\in\mathcal{F}\left(\mathcal{S}\right) .
\end{align}
\end{linenomath}
In the case that $S^{N}\times\mathcal{P}_{N}$ is finite, the singletons generate the canonical (i.e. discrete) $\sigma$-algebra on $S^{N}\times\mathcal{P}_{N}$, and the image $\Pi\left(s',\mathscr{P}'\right) =\left(s',\mathscr{P}'\right)\bmod\mathfrak{S}_{N}$ is measurable for each $\left(s',\mathscr{P}'\right)\in S^{N}\times\mathcal{P}_{N}$. In general, the image of a measurable set need not be measurable. However, the purpose of introducing an update pre-rule is to provide an alternative way to obtain an update rule. An update rule defines a Markov chain on the \textit{true} space of the evolutionary process, $\mathcal{S}$; pulling this chain back to $S^{N}\times\mathcal{P}_{N}$ is not necessary.
\end{remark}

An update pre-rule defines a Markov chain on $S^{N}\times\mathcal{P}_{N}$ whose kernel, $\kappa_{0}$, satisfies
\begin{linenomath}
\begin{align}
\kappa_{0}\Big( \left(s,\mathscr{P}\right) , E_{0} \Big) &:= \mathcal{U}_{0}\left[ F\Big( u\left(s,\mathscr{P}\right) \Big) \right]\Big( \left(s,\mathscr{P}\right) , E_{0} \Big)
\end{align}
\end{linenomath}
for each $\left(s,\mathscr{P}\right)\in S^{N}\times\mathcal{P}_{N}$ and $E_{0}\in\mathcal{F}\left(S^{N}\times\mathcal{P}_{N}\right)$. Denote by $\Pi_{\ast}\kappa_{0}$ the kernel of the Markov chain defined by $\Pi_{\ast}\mathcal{U}_{0}$. The stationary distribution(s) of $\kappa_{0}$ can be ``pushed forward" to stationary distribution(s) of $\Pi_{\ast}\kappa_{0}$ via the pushforward map of Eq. (\ref{eq:pushforwardOnMeasures}):

\begin{proposition}\label{prop:statDistPushforward}
If $\mu$ is a stationary distribution of the Markov chain defined by $\kappa_{0}$, then $\Pi_{\ast}\mu$ is a stationary distribution of the Markov chain defined by $\Pi_{\ast}\kappa_{0}$.
\end{proposition}
\begin{proof}
Suppose that $\mu$ is a stationary distribution of $\kappa_{0}$, i.e.
\begin{linenomath}
\begin{align}\label{eq:statDist}
\mu\left(E_{0}\right) &= \int_{\mathfrak{s}\in S^{N}\times\mathcal{P}_{N}}\kappa_{0}\left(\mathfrak{s},E_{0}\right)\,d\mu\left(\mathfrak{s}\right)
\end{align}
\end{linenomath}
for each $E_{0}\in\mathcal{F}\left(S^{N}\times\mathcal{P}_{N}\right)$. For each $E\in\mathcal{F}\left(\mathcal{S}\right)$, it follows that
\begin{linenomath}
\begin{align}
\int_{\mathfrak{s}\bmod\mathfrak{S}_{N}\in\mathcal{S}} &\left(\Pi_{\ast}\kappa_{0}\right)\left(\mathfrak{s}\bmod\mathfrak{S}_{N},E\right)\,d\left(\Pi_{\ast}\mu\right)\left(\mathfrak{s}\bmod\mathfrak{S}_{N}\right) \nonumber \\
&= \int_{\mathfrak{s}\in S^{N}\times\mathcal{P}_{N}}\left(\Pi_{\ast}\kappa_{0}\right)\left(\mathfrak{s}\bmod\mathfrak{S}_{N},E\right)\,d\mu\left(\mathfrak{s}\right) \nonumber \\
&= \int_{\mathfrak{s}\in S^{N}\times\mathcal{P}_{N}}\kappa_{0}\left(\mathfrak{s},\Pi^{-1}E\right)\,d\mu\left(\mathfrak{s}\right) \nonumber \\
&= \mu\left(\Pi^{-1}E\right) \nonumber \\
&= \left(\Pi_{\ast}\mu\right)\left(E\right)
\end{align}
\end{linenomath}
by the change of variables formula and Eq. (\ref{eq:statDist}), which completes the proof.
\end{proof}

Thus, the passage from an update pre-rule to an update rule via $\Pi_{\ast}$ is compatible with the steady states of the chains defined by $\kappa_{0}$ and $\Pi_{\ast}\kappa_{0}$, respectively.

\subsubsection{Examples}\label{subsubsec:URexamples}

We now give several classical examples of update pre-rules and rules:
\begin{example}[Death-birth process]\label{ex:dbExample}
Suppose that $S$ is finite and that $\mathcal{P}_{N}$ is the finite subset of $\mathcal{P}_{N}^{\textrm{G}}$ consisting of the undirected, unweighted graphs on $N$ vertices (with no other restrictions--this set contains regular graphs, scale-free networks, etc.). In each step of a death-birth process, a player is selected uniformly at random from the population for death. The neighbors (determined by a graph, $\Gamma\in\mathcal{P}_{N}$) then compete to fill the vacancy: a neighbor--say, player $j$--is chosen for reproduction with probability proportional to relative fitness, $x_{j}$. The offspring of this player inherits the strategy of the parent and fills the vacancy left by the deceased player. The population state (i.e. the graph) is left unchanged by this process. We define an update pre-rule for this process by giving transition probabilities from $S^{N}\times\mathcal{P}_{N}$ to itself when the $N-1$ surviving players in each round retain their labels. That is, if $\left(s,\Gamma\right)$ is the state of the process and player $i$ is chosen for death, $\Gamma$ remains the same and only the $i$th coordinate of $s$ is updated in order to obtain a new state, $\left(s',\Gamma\right)$. Thus, for two states, $\left(s,\Gamma\right) ,\left(s',\Gamma '\right)\in S^{N}\times\mathcal{P}_{N}$, it must be the case that $\Gamma =\Gamma '$ for there to be a nonzero probability of transitioning from $\left(s,\Gamma\right)$ to $\left(s',\Gamma '\right)$. The probability of choosing player $i$ for death is $1/N$, and, if this player is chosen for death, a transition is possible only if $s_{j}=s_{j}'$ for each $j\neq i$. The probability that player $i$ is replaced by the offspring of a player using $s_{i}$' is
\begin{linenomath}
\begin{align}
\sum_{j\neq i}\delta_{s_{j},s_{i}'}\left(\frac{\Gamma_{ji}x_{j}}{\sum_{j\neq i}\Gamma_{ji}x_{j}}\right) ,
\end{align}
\end{linenomath}
where $\delta_{s_{j},s_{i}'}$ is $1$ if $s_{j}=s_{i}'$ and $0$ otherwise. Thus, for $x\in\mathbb{R}^{N}$, the transition probability from $\left(s,\Gamma\right)$ to $\left(s',\Gamma '\right)$ is given by the update pre-rule, $\mathcal{U}_{0}$, defined by
\begin{linenomath}
\begin{align}
\mathcal{U}_{0}\left[x\right]\Big( \left(s,\Gamma\right) , \left(s',\Gamma '\right) \Big) &:= \delta_{\Gamma ,\Gamma '}\sum_{i=1}^{N}\left(\prod_{j\neq i}\delta_{s_{j},s_{j}'}\right)\left(\frac{1}{N}\right)\left(\frac{\sum_{j\neq i}\delta_{s_{j},s_{i}'}\Gamma_{ji}x_{j}}{\sum_{j\neq i}\Gamma_{ji}x_{j}}\right) .
\end{align}
\end{linenomath}
$\mathcal{U}_{0}$ is indeed an update pre-rule since, for each $\pi\in\mathfrak{S}_{N}$,
\begin{linenomath}
\begin{align}
\mathcal{U}_{0} &\left[\pi x\right]\Big( \left(\pi s,\pi\Gamma\right) , \left(\pi s',\pi\Gamma '\right) \Big) \nonumber \\
&= \delta_{\pi\Gamma ,\pi\Gamma '}\sum_{i=1}^{N}\left(\prod_{j\neq i}\delta_{s_{\pi\left(j\right)},s_{\pi\left(j\right)}'}\right)\left(\frac{1}{N}\right)\left(\frac{\sum_{j\neq i}\delta_{s_{\pi\left(j\right)},s_{\pi\left(i\right)}'}\Gamma_{\pi\left(j\right)\pi\left(i\right)}x_{\pi\left(j\right)}}{\sum_{j\neq i}\Gamma_{\pi\left(j\right)\pi\left(i\right)}x_{\pi\left(j\right)}}\right) \nonumber \\
&= \delta_{\Gamma ,\Gamma '}\sum_{i=1}^{N}\left(\prod_{j\neq\pi\left(i\right)}\delta_{s_{j},s_{j}'}\right)\left(\frac{1}{N}\right)\left(\frac{\sum_{j\neq\pi\left(i\right)}\delta_{s_{j},s_{\pi\left(i\right)}'}\Gamma_{j\pi\left(i\right)}x_{j}}{\sum_{j\neq\pi\left(i\right)}\Gamma_{j\pi\left(i\right)}x_{j}}\right) \nonumber \\
&= \delta_{\Gamma ,\Gamma '}\sum_{i=1}^{N}\left(\prod_{j\neq i}\delta_{s_{j},s_{j}'}\right)\left(\frac{1}{N}\right)\left(\frac{\sum_{j\neq i}\delta_{s_{j},s_{i}'}\Gamma_{ji}x_{j}}{\sum_{j\neq i}\Gamma_{ji}x_{j}}\right) \nonumber \\
&= \mathcal{U}_{0}\left[x\right]\Big( \left(s,\Gamma\right) , \left(s',\Gamma '\right) \Big) .
\end{align}
\end{linenomath}
This example verifies in detail the symmetry condition of an update pre-rule, Eq. (\ref{eq:updatePreRuleSymmetry}). Calculations for other processes are similar.
\end{example}

\begin{example}[Wright-Fisher process]\label{ex:wfExample}
In contrast to the Moran process of Example \ref{ex:moranProcess} and the death-birth process of Example \ref{ex:dbExample}, one could also consider a process in which the entire population is updated synchronously. For example, in the \textit{Wright-Fisher process} \citep{ewens:S:2004,imhof:JMB:2006}, the population is updated as follows: A player--say, player $i$--is first selected for reproduction with probability proportional to fitness. The offspring of this player inherits the strategy of the parent with probability $1-\varepsilon_{i}$ and takes on a novel strategy uniformly at random with probability $\varepsilon_{i}$. The mutation rate of the offspring is inherited from the parent (so that the offspring's offspring will also mutate with probability $\varepsilon_{i}$). This process is then repeated until there are $N$ new offspring, and these offspring constitute the new population. Thus, one update step of the Wright-Fisher process involves updating the entire population.

Let $S$ be finite and let $\mathfrak{m}$ be some finite subset of $\left[0,1\right]$. The population state space for this version of the Wright-Fisher process is $\mathcal{P}_{N}:=\mathfrak{m}^{N}$. That is, a population state is an $N$-tuple of strategy-mutation rates, $\varepsilon$, with $\varepsilon_{i}$ the mutation rate for player $i$. We define an update pre-rule, $\mathcal{U}_{0}$, as follows: for $x\in\mathbb{R}^{N}$ and $\left(s,\varepsilon\right) , \left(s',\varepsilon '\right)\in S^{N}\times\mathcal{P}_{N}$, let
\begin{linenomath}
\begin{align}
\mathcal{U}_{0}\left[x\right]\Big( \left(s,\varepsilon\right) , \left(s',\varepsilon '\right) \Big) &= \prod_{i=1}^{N}\sum_{j=1}^{N}\delta_{\varepsilon_{i}',\varepsilon_{j}}\left(\frac{x_{j}}{x_{1}+\cdots +x_{N}}\right)\left[\delta_{s_{i}',s_{j}}\left(1-\varepsilon_{j}\right) + \varepsilon_{j}\left(\frac{1}{n}\right)\right] .
\end{align}
\end{linenomath}
For each $\pi ,\tau\in\mathfrak{S}_{N}$, $x\in\mathbb{R}^{N}$, and $\left(s,\varepsilon\right) ,\left(s',\varepsilon '\right)\in S^{N}\times\mathcal{P}_{N}$, it is readily verified that
\begin{linenomath}
\begin{align}
\mathcal{U}_{0}\left[\pi x\right]\Big( \left(\pi s,\pi\varepsilon\right) , \left(\tau s',\tau\varepsilon '\right) \Big) &= \mathcal{U}_{0}\left[x\right]\Big( \left(s,\varepsilon\right) , \left(s',\varepsilon '\right) \Big) .
\end{align}
\end{linenomath}
Therefore, the resulting update rule, $\mathcal{U}:=\Pi_{\ast}\mathcal{U}_{0}$, satisfies
\begin{linenomath}
\begin{align}
\mathcal{U} &\left[x\right]\Big( \left(s,\varepsilon\right)\bmod\mathfrak{S}_{N} , \left(s',\varepsilon '\right)\bmod\mathfrak{S}_{N} \Big) \nonumber \\
&= \left|\textrm{orb}_{\mathfrak{S}_{N}}\left(s',\varepsilon '\right)\right|\prod_{i=1}^{N}\sum_{j=1}^{N}\delta_{\varepsilon_{i}',\varepsilon_{j}}\left(\frac{x_{j}}{x_{1}+\cdots +x_{N}}\right)\left[\delta_{s_{i}',s_{j}}\left(1-\varepsilon_{j}\right) + \varepsilon_{j}\left(\frac{1}{n}\right)\right] .
\end{align}
\end{linenomath}
Consider the simple case $\varepsilon =\varepsilon '=\mathbf{0}$ (meaning there are no strategy mutations). If $k_{r}'$ denotes the frequency of strategy $r$ in state $s'$ for $r=1,\dots ,n$, then $\left|\textrm{Stab}_{\mathfrak{S}_{N}}\left(s',\mathbf{0}\right)\right| =k_{1}'!\cdots k_{n}'!$, which means that
\begin{linenomath}
\begin{align}
\mathcal{U} &\left[x\right]\Big( \left(s,\mathbf{0}\right)\bmod\mathfrak{S}_{N} , \left(s',\mathbf{0}\right)\bmod\mathfrak{S}_{N} \Big) \nonumber \\
&= \left|\textrm{orb}_{\mathfrak{S}_{N}}\left(s',\mathbf{0}\right)\right|\prod_{i=1}^{N}\sum_{j=1}^{N}\left(\frac{x_{j}}{x_{1}+\cdots +x_{N}}\right)\delta_{s_{i}',s_{j}} \nonumber \\
&= \frac{\left|\mathfrak{S}_{N}\right|}{\left|\textrm{Stab}_{\mathfrak{S}_{N}}\left(s',\mathbf{0}\right)\right|}\prod_{i=1}^{N}\sum_{j=1}^{N}\left(\frac{x_{j}}{x_{1}+\cdots +x_{N}}\right)\delta_{s_{i}',s_{j}} \nonumber \\
&= \binom{N}{k_{1}',\dots ,k_{n}'}\prod_{i=1}^{N}\sum_{j=1}^{N}\left(\frac{x_{j}}{x_{1}+\cdots +x_{N}}\right)\delta_{s_{i}',s_{j}} , \label{eq:classicalWF}
\end{align}
\end{linenomath}
where the third line was obtained using the orbit-stabilizer theorem \citep[see][]{knapp:BB:2006}. Eq. (\ref{eq:classicalWF}) is just the classical formula for the transition probabilities of the Wright-Fisher process based on multinomial sampling \citep{kingman:SIAM:1980,durrett:SNY:2002,der:TPB:2011}.
\end{example}

\begin{example}[Pairwise comparison process]\label{ex:pcExample}
In each of our examples so far, both $S$ and $\mathcal{P}_{N}$ have been finite. Since our theory allows for these sets to be measurable, we now give an example of an evolutionary process whose strategy space is continuous. Let $S=\left[0,K\right]$ for some $K>0$. This interval might be the strategy space for a public goods game, for instance, with $K$ the maximum amount any one player may contribute to the public good. As in Examples \ref{ex:moranProcess} and \ref{ex:wfExample}, let $\mathcal{P}_{N}:=\mathfrak{m}^{N}$ for some finite subset, $\mathfrak{m}$, of $\left[0,1\right]$; an element of $\mathcal{P}_{N}$ is just a profile of mutation rates, $\varepsilon$.

In each step of a pairwise comparison process, a player--say, player $i$--is selected uniformly at random from the population to evaluate his or her strategy. Another player--say, player $j$--is then chosen uniformly at random from the rest of the population as a model player. With probability $1-\varepsilon_{i}$, the focal player takes into account the model player and probabilistically updates his or her strategy as follows: if $x_{i}$ (resp. $x_{j}$) is the fitness of the focal (resp. model) player, and if $\beta\geqslant 0$ is the selection intensity, then the focal player imitates the model player with probability
\begin{linenomath}
\begin{align}
\frac{1}{1+e^{-\beta\left(x_{j}-x_{i}\right)}}
\end{align}
\end{linenomath}
and retains his or her current strategy with probability
\begin{linenomath}
\begin{align}
\frac{1}{1+e^{-\beta\left(x_{i}-x_{j}\right)}}
\end{align}
\end{linenomath}
\citep{szabo:PRE:1998}. On the other hand, with probability $\varepsilon_{i}$ the focal player ignores the model player completely. In this case, the focal player ``explores" and adopts a new strategy from the interval $\left[0,K\right]$ probabilistically according to a truncated Gaussian distribution centered at $s_{i}$ (the current strategy of the focal player). For some specified variance, $\sigma^{2}$, this truncated Gaussian distribution has for a density function
\begin{linenomath}
\begin{align}\label{truncatedGaussianDensity}
\phi_{s_{i}}\left(x\right) &:= \left(\int_{0}^{K}\exp\left(-\frac{\left(y-s_{i}\right)^{2}}{2\sigma^{2}}\right)\,dy\right)^{-1}\exp\left(-\frac{\left(x-s_{i}\right)^{2}}{2\sigma^{2}}\right) .
\end{align}
\end{linenomath}
The parameter $\sigma$ may be interpreted as a measure of how venturesome a player is, with cautious exploration corresponding to small $\sigma$ and risky exploration corresponding to large $\sigma$. The density function, $\phi_{s_{i}}$, defines a probability measure,
\begin{linenomath}
\begin{align}
\Phi_{s_{i}} &: \mathcal{F}\left(S\right) \longrightarrow \left[0,1\right] \nonumber \\
&: E \longmapsto \int_{E}\phi_{s_{i}}\left(x\right)\,dx .
\end{align}
\end{linenomath}
If player $i$ ignores the model player, then he or she adopts a strategy from $E\in\mathcal{F}\left(S\right)$ with probability $\Phi_{s_{i}}\left(E\right)$. Thus, a player who explores is more likely to adopt a strategy close to his or her current strategy than one farther away. For $\beta\geqslant 0$, let
\begin{linenomath}
\begin{align}
g_{\beta}\left(x\right) &:= \frac{1}{1+e^{-\beta x}}
\end{align}
\end{linenomath}
be the logistic function. (In terms of this function, the probability that a focal player with fitness $x_{i}$ imitates a model player with fitness $x_{j}$ is $g_{\beta}\left(x_{j}-x_{i}\right)$.) We assemble these components into an update pre-rule, $\mathcal{U}_{0}$, as follows: For $x\in\mathbb{R}^{N}$, $\left(s,\varepsilon\right)\in S^{N}\times\mathcal{P}_{N}$, and a measurable rectangle, $E_{1}\times\cdots\times E_{N}\times E'\in\mathcal{F}\left(S\right)^{N}\times\mathcal{F}\left(\mathcal{P}_{N}\right)$, let
\begin{linenomath}
\begin{align}
\mathcal{U}_{0} &\left[x\right]\Big( \left(s,\varepsilon\right) , E_{1}\times\cdots\times E_{N}\times E' \Big) \nonumber \\
&:= \delta_{\varepsilon}\left(E'\right)\sum_{i=1}^{N}\left(\frac{1}{N}\right)\left(\prod_{j\neq i}\delta_{s_{j}}\left(E_{j}\right)\right)\sum_{j\neq i}\left(\frac{1}{N-1}\right)\Bigg\{\varepsilon_{i}\Phi_{s_{i}}\left(E_{i}\right)  \nonumber \\
&\quad\quad +\left(1-\varepsilon_{i}\right)\Big[\delta_{s_{j}}\left(E_{i}\right) g_{\beta}\left(x_{j}-x_{i}\right) + \delta_{s_{i}}\left(E_{i}\right) g_{\beta}\left(x_{i}-x_{j}\right)\Big]\Bigg\} ,
\end{align}
\end{linenomath}
and extend this definition additively to disjoint unions of measurable rectangles. For each $x\in\mathbb{R}^{N}$ and $\left(s,\varepsilon\right)\in S^{N}\times\mathcal{P}_{N}$, one can verify that $\mathcal{U}_{0}\left[x\right]\Big(\left(s,\varepsilon\right) ,\--\Big)$ extends to a measure on $S^{N}\times\mathcal{P}_{N}$ by the Hahn-Kolmogorov theorem, which we also denote by $\mathcal{U}_{0}\left[x\right]\Big(\left(s,\varepsilon\right) ,\--\Big)$. It is readily verified that $\mathcal{U}_{0}$ is an update pre-rule, so $\mathcal{U}_{0}$ extends to an update rule, $\Pi_{\ast}\mathcal{U}_{0}$, by Proposition \ref{prop:preRuleToRule}. This example illustrates how the strategy mutations might themselves depend on the strategies (as opposed to simply being uniform random variables on $S$ as they were in the previous examples).
\end{example}

In Examples \ref{ex:moranProcess}, \ref{ex:wfExample}, and \ref{ex:pcExample}, we considered processes with heterogeneous mutation rates (meaning $\varepsilon_{i}$ depends on $i$). In Examples \ref{ex:moranProcess} and \ref{ex:wfExample}, there is a nonzero probability of transitioning from a state with heterogeneous mutation rates to a state with homogeneous mutation rates. For instance, in the Wright-Fisher process of Example \ref{ex:wfExample}, if $\left(s,\varepsilon\right)$ is a state such that $\varepsilon_{\ell}\neq\varepsilon_{\ell '}$ for some $\ell$ and $\ell '$, and if $\left(s',\varepsilon '\right)$ is a state satisfying $s_{1}'=s_{2}'=\cdots =s_{N}'=s_{\ell}$ and $\varepsilon_{1}'=\varepsilon_{2}'=\cdots =\varepsilon_{N}'=\varepsilon_{\ell}$, then there is a nonzero probability of transitioning from $\left(s,\varepsilon\right)$ to $\left(s',\varepsilon '\right)$ provided $x_{\ell}>0$. Thus, the population state, which is simply of profile of mutation rates, can change from generation to generation. In contrast, the population state of Example \ref{ex:pcExample} \textit{cannot} change from generation to generation since strategies are imitated and mutation rates are not inherited. In other words, much of the biological meaning behind the quantities appearing in the population state are encoded in the dynamics of the process via the update rule.

In a more formal setting, let $\kappa$ be the transition kernel obtained from an update rule via Eq. (\ref{eq:updateKernel}). From the projection $\Pi_{2}:S^{N}\times\mathcal{P}_{N}\rightarrow\mathcal{P}_{N}$, we obtain a map
\begin{linenomath}
\begin{align}
\widetilde{\Pi}_{2} &: \left(S^{N}\times\mathcal{P}_{N}\right) /\mathfrak{S}_{N} \longrightarrow \mathcal{P}_{N}/\mathfrak{S}_{N} \nonumber \\
&: \left(s,\mathscr{P}\right)\bmod\mathfrak{S}_{N} \longmapsto \mathscr{P}\bmod\mathfrak{S}_{N} .
\end{align}
\end{linenomath}
$\widetilde{\Pi}_{2}$ gives us a pushforward map, $\left(\widetilde{\Pi}_{2}\right)_{\ast}:\Delta\left(\mathcal{S}\right)\rightarrow\Delta\left(\mathcal{P}_{N}/\mathfrak{S}_{N}\right)$, which we use to formalize the intuition behind ``static" and ``dynamic" population states:
\begin{definition}[Static and dynamic population states]
A population state, $\mathscr{P}$, in a population state space, $\mathcal{P}_{N}$, is \textit{static} relative to $\kappa$ if, for each $s\in S^{N}$,
\begin{linenomath}
\begin{align}
\left(\widetilde{\Pi}_{2}\right)_{\ast}\kappa\Big(\left(s,\mathscr{P}\right)\bmod\mathfrak{S}_{N},\--\Big) &= \delta_{\mathscr{P}\bmod\mathfrak{S}_{N}} ,
\end{align}
\end{linenomath}
where $\delta_{\mathscr{P}\bmod\mathfrak{S}_{N}}$ denotes the Dirac measure on $\Delta\left(\mathcal{P}_{N}/\mathfrak{S}_{N}\right)$ centered at $\mathscr{P}\bmod\mathfrak{S}_{N}$. Otherwise, if $\kappa$ is not static relative to $\kappa$, we say that $\mathscr{P}$ is \textit{dynamic} relative to $\kappa$.
\end{definition}

In Examples \ref{ex:dbExample} and \ref{ex:pcExample}, every population state is static. In Examples \ref{ex:moranProcess} and \ref{ex:wfExample}, only the population states (i.e. mutation profiles) with $\varepsilon_{1}=\varepsilon_{2}=\cdots =\varepsilon_{N}$ are static.

\section{Stochastic selection processes with variable population size}\label{subsec:variablePop}

Suppose now that the population size is dynamic, and let $\mathcal{N}\subseteq\left\{0,1,2,\dots\right\}$ be the set of admissible population sizes. As in \S\ref{sec:ssp}, let $S$ be the strategy space for each player. Instead of having a single population state space, we now require the existence of a population state space, $\mathcal{P}_{N}$, for \textit{each} $N\in\mathcal{N}$. The state space for such a process is
\begin{linenomath}
\begin{align}
\mathcal{S} &:= \bigsqcup_{\ell\in\mathcal{N}} \left(S^{\ell}\times\mathcal{P}_{\ell}\right) /\mathfrak{S}_{\ell} ,
\end{align}
\end{linenomath}
where $\bigsqcup_{\ell\in\mathcal{N}} \left(S^{\ell}\times\mathcal{P}_{\ell}\right) /\mathfrak{S}_{\ell}$ denotes the disjoint union of the spaces $\left(S^{\ell}\times\mathcal{P}_{\ell}\right) /\mathfrak{S}_{\ell}$. Instead of a single aggregate payoff function and update rule, we now require that there be an aggregate payoff function, $u^{N}:S^{N}\times\mathcal{P}_{N}\rightarrow\mathbb{R}^{N}$, and an update rule, $\mathcal{U}^{N}:\mathbb{R}^{N}\rightarrow K\left(S^{N}\times\mathcal{P}_{N},\mathcal{S}\right)$, for \textit{each} admissible population size, $N\in\mathcal{N}$. If the population currently has size $N$, then $u^{N}$ determines the payoffs to the players in the interaction step, and $\mathcal{U}^{N}$ updates the population (possibly to one of a different size). Of course, for each $\pi\in\mathfrak{S}_{N}$, $x\in\mathbb{R}^{N}$, $s\in S^{N}$, $\mathscr{P}\in\mathcal{P}_{N}$, and $E\in\mathcal{F}\left(\mathcal{S}\right)$, these functions must satisfy
\begin{linenomath}
\begin{subequations}
\begin{align}
u^{N}\left(\pi s,\pi\mathscr{P}\right) &= \pi u^{N}\left(s,\mathscr{P}\right) ; \label{eq:ellPayoffSymmetry} \\
\mathcal{U}^{N}\left[\pi x\right]\Big( \left(\pi s,\pi\mathscr{P}\right) , E \Big) &= \mathcal{U}^{N}\left[x\right]\Big( \left(s,\mathscr{P}\right) , E \Big) \label{eq:ellUpdateSymmetry}
\end{align}
\end{subequations}
\end{linenomath}
just as they did in Definitions \ref{def:aggregatePayoff} and \ref{def:updateRule}, respectively.

Finally, we have the definition of a stochastic selection process in its full generality:
\begin{definition}[Stochastic selection process]
A \textit{stochastic selection process} consists of the following components:
\begin{enumerate}

\item[(1)] a set of admissible population sizes, $\mathcal{N}\subseteq\left\{0,1,2,\dots\right\}$;

\item[(2)] a measurable strategy space, $S$;

\item[(3)] for each $N\in\mathcal{N}$, a population state space, $\mathcal{P}_{N}$;

\item[(4)] for each $N\in\mathcal{N}$, an aggregate payoff function, $u^{N}:S^{N}\times\mathcal{P}_{N}\rightarrow\mathbb{R}^{N}$;

\item[(5)] a payoff-to-fitness function, $f:\mathbb{R}\rightarrow\mathbb{R}$;

\item[(6)] for each $N\in\mathcal{N}$, an update rule, $\mathcal{U}^{N}:\mathbb{R}^{N}\rightarrow K\left(S^{N}\times\mathcal{P}_{N},\mathcal{S}\right)$, where
\begin{linenomath}
\begin{align}
\mathcal{S} &:= \bigsqcup_{\ell\in\mathcal{N}} \left(S^{\ell}\times\mathcal{P}_{\ell}\right) /\mathfrak{S}_{\ell} .
\end{align}
\end{linenomath}
\end{enumerate}
\end{definition}

The components of a stochastic selection process produce a Markov chain on $\mathcal{S}$ whose kernel, $\kappa$, is defined as follows: for $N\in\mathcal{N}$, $\left(s,\mathscr{P}\right)\in S^{N}\times\mathcal{P}_{N}$, and $E\in\mathcal{S}$,
\begin{linenomath}
\begin{align}
\kappa\Big( \left(s,\mathscr{P}\right)\bmod\mathfrak{S}_{N} , E \Big) &= \mathcal{U}^{N}\left[F\Big(u^{N}\left(s,\mathscr{P}\right)\Big)\right]\Big( \left(s,\mathscr{P}\right) , E \Big) .
\end{align}
\end{linenomath}
It is readily verified that $\kappa$ is well defined (see Eq. (\ref{eq:kappaWellDefined})).

\begin{remark}
The payoff-to-fitness function, $f$, in requirement (5) of a stochastic selection process, is not strictly necessary. It could instead be absorbed into either the payoff function (which would then be a \textit{fitness function}) or the update rule (which would then be a family of transition kernels parametrized by \textit{payoff profiles} rather than by fitness profiles). We include this function as a part of a stochastic selection process for three reasons: (1) having an aggregate payoff function instead of an aggregate fitness function allows for a more straightforward comparison to the theory of stochastic games; (2) having an update rule be a family of transition kernels parametrized by fitness simplifies its presentation (see \S\ref{subsubsec:URexamples}); and (3) payoff-to-fitness functions are often explicitly mentioned in models of evolutionary games in the literature. Tuning the selection strength of a process, for instance, amounts to modifying the payoff-to-fitness function, so including this function in a stochastic selection process allows one to more explicitly separate the various components of a selection process.
\end{remark}

\begin{remark}
The notion of \textit{update pre-rule} also makes sense for populations of variable size, although one must define the symmetry condition, Eq. (\ref{eq:updatePreRuleSymmetry}), with greater care. If $\mathcal{N}$ is the set of admissible population sizes, then we require--for each $N\in\mathcal{N}$--a map
\begin{linenomath}
\begin{align}
\mathcal{U}_{0}^{N} &: \mathbb{R}^{N} \longrightarrow K\left(S^{N}\times\mathcal{P}_{N},\bigsqcup_{\ell\in\mathcal{N}}S^{\ell}\times\mathcal{P}_{\ell}\right) .
\end{align}
\end{linenomath}
For each $N\in\mathcal{N}$, there is an action of $\mathfrak{S}_{N}$ on $\bigsqcup_{\ell\in\mathcal{N}}S^{\ell}\times\mathcal{P}_{\ell}$ defined by
\begin{linenomath}
\begin{align}
\pi\left(s,\mathscr{P}\right) &= \begin{cases}\left(\pi s,\pi\mathscr{P}\right) & \left(s,\mathscr{P}\right)\in S^{N}\times\mathcal{P}_{N} ; \\ \left(s,\mathscr{P}\right) & \left(s,\mathscr{P}\right)\not\in S^{N}\times\mathcal{P}_{N} .\end{cases} \label{eq:extendedAction}
\end{align}
\end{linenomath}
From the set of groups $\left\{ \mathfrak{S}_{N} \right\}_{N\in\mathcal{N}}$, one can construct the \textit{free product},
\begin{linenomath}
\begin{align}\label{eq:freeProduct}
\mathfrak{S}_{\mathcal{N}} &:= \bigast_{N\in\mathcal{N}}\mathfrak{S}_{N} ,
\end{align}
\end{linenomath}
which is just the analogue of disjoint union in the category of groups \citep[see][]{knapp:BB:2006}. Collectively, the actions of $S^{N}\times\mathcal{P}_{N}$ on $\bigsqcup_{\ell\in\mathcal{N}}S^{\ell}\times\mathcal{P}_{\ell}$ defined by Eq. (\ref{eq:extendedAction}) (over all $N\in\mathcal{N}$) result in a (measurable) action of $\mathfrak{S}_{\mathcal{N}}$ on $\bigsqcup_{\ell\in\mathcal{N}}S^{\ell}\times\mathcal{P}_{\ell}$. For $\left\{\mathcal{U}_{0}^{N}\right\}_{N\in\mathcal{N}}$ to define a collection of update pre-rules, we require that for each $N\in\mathcal{N}$, $\pi\in\mathfrak{S}_{N}$, $x\in\mathbb{R}^{N}$, $\left(s,\mathscr{P}\right)\in S^{N}\times\mathcal{P}_{N}$ and $E_{0}\in\mathcal{F}\left(\bigsqcup_{\ell\in\mathcal{N}}S^{\ell}\times\mathcal{P}_{\ell}\right)$, there exists $\tau\in\mathfrak{S}_{\mathcal{N}}$ such that
\begin{linenomath}
\begin{align}
\mathcal{U}_{0}\left[\pi x\right]\Big( \left(\pi s,\pi\mathscr{P}\right) , E_{0} \Big) &= \mathcal{U}_{0}\left[x\right]\Big( \left(s,\mathscr{P}\right) , \tau E_{0} \Big) .
\end{align}
\end{linenomath}
The reason we require $\tau$ to be in the free product, $\mathfrak{S}_{\mathcal{N}}$, instead of just in $\mathfrak{S}_{N}$ for some $N\in\mathcal{N}$, is that it need not hold that $E_{0}\in\mathcal{F}\left(S^{N}\times\mathcal{P}_{N}\right)$ for some $N\in\mathcal{N}$. $E_{0}$ could be some complicated measurable set consisting of elements of $S^{N}\times\mathcal{P}_{N}$ for \textit{several} $N\in\mathcal{N}$, so we need some way of relabeling elements $S^{N}\times\mathcal{P}_{N}$ for several values of $N\in\mathcal{N}$ simultaneously. Extending the action of $\mathfrak{S}_{N}$ on $S^{N}\times\mathcal{P}_{N}$ to $\bigsqcup_{\ell\in\mathcal{N}}S^{\ell}\times\mathcal{P}_{\ell}$ via (\ref{eq:extendedAction}), in order to form the free product, $\mathfrak{S}_{\mathcal{N}}$, via Eq. (\ref{eq:freeProduct}), accomplishes this task.

Analogues of Propositions \ref{prop:preRuleToRule}, \ref{prop:updatePullback}, and \ref{prop:statDistPushforward} also hold in this context, but we do not go through the details here; the proofs are essentially the same as they were in \S\ref{subsubsec:updatePreRules}.
\end{remark}

\section{Discussion}\label{sec:discussion}

We use the term ``selection process" instead of ``evolutionary game" in order to emphasize that the update step is based on the principles of natural selection. Several types of adaptive processes appearing in the economics literature have been referred to as evolutionary games. \textit{Best-response dynamics} of \citet{ellison:E:1993} is a procedure in which, at each round, the players update their strategies based on the best responses to their opponents in the previous round. This process is known to converge to a Nash equilibrium of the game. \citet{hart:E:2000} define a similar process called \textit{regret matching} that leads to a correlated equilibrium of the game. These processes can be phrased as stochastic games (along with appropriate strategies), but they are \textit{not} stochastic selection processes, as we now illustrate with best-response dynamics:

Suppose $N=2$. Let $S=\left\{A,B\right\}$ and let $u:S^{2}\rightarrow\mathbb{R}^{2}$ be the payoff function for a game between two players. If best-response dynamics in this population defines a stochastic selection process, then there exists a population state space, $\mathcal{P}_{2}$, and an update rule, $\mathcal{U}$, such that the transition kernel of the resulting Markov chain, $\kappa_{u}$, satisfies
\begin{linenomath}
\begin{align}
\kappa_{u}\Big( \left(s,\mathscr{P}\right) ,E\Big) &= \mathcal{U}\left[u\left(s\right)\right]\Big(\left(s,\mathscr{P}\right) ,E\Big)
\end{align}
\end{linenomath}
for each $\left(s,\mathscr{P}\right)\in S^{2}\times\mathcal{P}_{2}$ and $E\in\mathcal{F}\left(\mathcal{S}\right)$. In other words, $\kappa_{u}$ depends on $u$ insofar as $\mathcal{U}$ depends on $\mathbb{R}^{2}$. Consider the two payoff functions, $u,v:S^{2}\rightarrow\mathbb{R}^{2}$, defined by
\begin{linenomath}
\begin{subequations}
\begin{align}
\begin{pmatrix}u\left(A,A\right) & u\left(A,B\right) \\ u\left(B,A\right) & u\left(B,B\right)\end{pmatrix} &= \begin{pmatrix}2 & 2 \\ 1 & 1\end{pmatrix} ; \\
\begin{pmatrix}v\left(A,A\right) & v\left(A,B\right) \\ v\left(B,A\right) & v\left(B,B\right)\end{pmatrix} &= \begin{pmatrix}2 & 0 \\ 1 & 1\end{pmatrix} .
\end{align}
\end{subequations}
\end{linenomath}
Since $u\left(B,B\right) =v\left(B,B\right) =1$, it follows that for $s=\left(B,B\right)$ and any $\mathscr{P}\in\mathcal{P}_{2}$,
\begin{linenomath}
\begin{align}
\kappa_{u}\left( \Big(\left(B,B\right) , \mathscr{P}\Big) , E \right) &= \mathcal{U}\left[u\left(B,B\right)\right]\left(\Big(\left(B,B\right) ,\mathscr{P}\Big) ,E\right) \nonumber \\
&= \mathcal{U}\left[v\left(B,B\right)\right]\left(\Big(\left(B,B\right) ,\mathscr{P}\Big) ,E\right) \nonumber \\
&= \kappa_{v}\left( \Big(\left(B,B\right) , \mathscr{P}\Big) , E \right) .
\end{align}
\end{linenomath}
Thus, the strategy profile $\left(B,B\right)$ must be updated by best-response dynamics in the same way for both functions, $u$ and $v$. However, best-response dynamics actually results in different updates of $\left(B,B\right)$ for these two games: for $u$, the profile $\left(B,B\right)$ is updated to $\left(A,A\right)$; for $v$, the profile $\left(B,B\right)$ is updated to $\left(B,B\right)$ (it is already a Nash equilibrium). The key observation is that, while the Markov chain defined by a stochastic selection process depends on the aggregate payoff function, $u$, the update rule is \textit{independent} of $u$. In contrast, the update step in best-response dynamics clearly depends on $u$. In a stochastic selection process, the only role of the aggregate payoff function is to determine the fitness profile (which is then passed to the update rule).

An update rule in the classical sense, vaguely speaking, generally consists of information about births and deaths (or imitation). Choosing to represent the strategies of the players in the population as an $N$-tuple, $s\in S^{N}$, is just a mathematical convenience; the update rule is independent of how the players are labeled. We defined the notion of update pre-rule (Definition \ref{def:updatePreRule}) in order to relate stochastic selection processes to the way in which evolutionary processes are frequently modeled--as Markov chains on $S^{N}$ (or, more generally, on $S^{N}\times\mathcal{P}_{N}$). In many cases, an update pre-rule is simpler to write down explicitly than an update rule since one can choose a convenient enumeration of the players in each time step. We showed that an update pre-rule can always be ``pushed forward" to an update rule, so it is sufficient to give an update pre-rule in place of an update rule in the definition of stochastic selection process. However, an update \textit{rule} is the true mathematical formalization of an evolutionary update in this context since it is independent of the labeling of the players.

Stochastic selection processes encompass existing models of selection such as the \textit{evolutionary game Markov chain} of \citet{wage:thesis:2010} and the \textit{evolutionary Markov chain} of \citet{allen:JMB:2012}. \citet{wage:thesis:2010} considers a Markov chain arising from probability distributions over a collection of \textit{inheritance rules}. An inheritance rule is a map, $I:\left\{1,\dots ,N\right\}\rightarrow\left\{1,\dots ,N,m\right\}$, that designates the source of a player's strategy: if $I\left(i\right) =j\neq m$, then player $i$ inherits his or her strategy from player $j$; if $I\left(i\right) =m$, then player $i$'s strategy is the result of a random mutation (assumed to be uniform over the strategy set). Similarly, \citet{allen:JMB:2012} model evolution in populations with fixed size and structure using \textit{replacement events}. A replacement event is a pair, $\left(R,\alpha\right)$, consisting of a collection, $R\subseteq\left\{1,\dots ,N\right\}$, of players who are replaced and a rule, $\alpha :R\rightarrow\left\{1,\dots ,N\right\}$, indicating the parent of each offspring involved in the replacement. These frameworks provide good models for many classical evolutionary processes, but they do not account for genetic processes with crossover, for instance. Moreover, one could imagine a cultural process in which a player updates his or her strategy based on some complicated synthesis of \textit{many} strategies in the population. Our framework generalizes these models, taking into account arbitrary strategy spaces, payoff functions, population structures, and update rules.

\citet{mcavoy:JRSI:2015} define the notion of a \textit{homogeneous} evolutionary game, which, informally, means that any two states consisting of a single $A$-mutant in a monomorphic $B$-population are equivalent. More specifically, suppose that a population state consists of a graph and a profile of mutation rates; that is, $\mathcal{P}_{N}=\mathcal{P}_{N}^{\textrm{G}}\times\mathfrak{m}^{N}$ for some finite subset, $\mathfrak{m}$, of $\left[0,1\right]$. If $s,s'\in S^{N}$, and if $\pi\in\mathfrak{S}_{N}$ satisfies $\pi\Gamma =\Gamma$ and $\pi\varepsilon =\varepsilon$ for some $\Gamma\in\mathcal{P}_{N}^{\textrm{G}}$ and $\varepsilon\in\mathfrak{m}^{N}$, then \citet{mcavoy:JRSI:2015} argue that the states $\left(s,\left(\Gamma ,\varepsilon\right)\right)$ and $\left(\pi s,\left(\Gamma ,\varepsilon\right)\right)$ are \textit{evolutionarily equivalent} in the sense that $\pi$ induces an automorphism of the Markov chain on $S^{N}\times\mathcal{P}_{N}$ that sends $\left(s,\left(\Gamma ,\varepsilon\right)\right)$ to $\left(\pi s,\left(\Gamma ,\varepsilon\right)\right)$. This result is a special case of an observation that is completely obvious in the context of stochastic selection processes: if $\pi\in\textrm{Aut}\left(\mathscr{P}\right)$ for some $\mathscr{P}$ in a population state space, $\mathcal{P}_{N}$, then the representatives $\left(s,\mathscr{P}\right)$ and $\left(\pi s,\mathscr{P}\right) =\left(\pi s,\pi\mathscr{P}\right)$ define \textit{exactly the same state} in $\mathcal{S}$. Therefore, working on the true state space of an evolutionary process, $\mathcal{S}$, helps to elucidate structural symmetries that are not as clear when working with a Markov chain on $S^{N}\times\mathcal{P}_{N}$ that is defined by an update pre-rule.

In the definition of a population state space (Definition \ref{def:structureSpace}), we require a measurable space, $\mathcal{P}_{N}$, along with a measurable action of $\mathfrak{S}_{N}$ on $\mathcal{P}_{N}$. Naturally, this setup raises the question of what types of $\mathfrak{S}_{N}$-actions one can put on $\mathcal{P}_{N}$ while still retaining the desired properties of an evolutionary process. If one were to take an update rule, $\mathcal{U}$, and then arbitrarily change the action of $\mathfrak{S}_{N}$ on $\mathcal{P}_{N}$, then $\mathcal{U}$ need not remain an update rule under this new action. For example, let $\Gamma^{0}\in\mathcal{P}_{N}^{\textrm{G}}$ be the Frucht graph, which is an undirected, unweighted, regular graph with $N=12$ vertices and no nontrivial symmetries \citep{frucht:CM:1939}. In other words, if $\Gamma_{\pi\left(i\right)\pi\left(j\right)}^{0}=\Gamma_{ij}^{0}$ for each $i,j=1,\dots ,12$, then $\pi =\textrm{id}$. Instead of the standard action of $\mathfrak{S}_{N}$ on $\mathcal{P}_{N}^{\textrm{G}}$, one could instead declare that $\mathfrak{S}_{N}$ acts trivially on $\mathcal{P}_{N}^{\textrm{G}}$; in particular, $\pi\star\Gamma =\Gamma$ for each $\pi\in\mathfrak{S}_{N}$ and $\Gamma\in\mathcal{P}_{N}^{\textrm{G}}$. If $\mathcal{U}$ is the update rule for the death-birth process (see Example \ref{ex:dbExample}), then it follows that the representatives $\left(s,\Gamma^{0}\right)$ and $\left(\pi s,\pi\star\Gamma^{0}\right) =\left(\pi s,\Gamma^{0}\right)$ define the same point in the state space, $\mathcal{S}$. For the (frequency-dependent) Snowdrift Game, all $12$ states consisting of a single cooperator in a population of defectors give rise to different fixation probabilities for cooperators \citep[see][]{mcavoy:JRSI:2015}. Therefore, it cannot be the case that $\mathcal{U}$ defines a Markov chain on $\mathcal{S}$ via Eq. (\ref{eq:updateKernel}); in particular, $\mathcal{U}$ is no longer a well-defined update rule under the new action, $\star$, of $\mathfrak{S}_{N}$ on $\mathcal{P}_{N}$.

The dynamics of a stochastic selection process, which are obtained via its update rule, encode much of the biological meaning of the components that constitute the population state space. In both of Examples \ref{ex:wfExample} and \ref{ex:pcExample}, the population state space was $\mathcal{P}_{N}:=\mathfrak{m}^{N}$ for some finite subset, $\mathfrak{m}$, of $\left[0,1\right]$. On the other hand, the interpretations of these mutation rates were completely different in these two processes: in Example \ref{ex:wfExample}, the mutation applied to the offspring of reproducing players; in Example \ref{ex:pcExample}, the mutation was interpreted as ``exploration" and applied to a player who was chosen to update his or her strategy. One could even consider different implementations of mutation rates in the same process: in a genetic process based on reproduction, a player's strategy-mutation rate may be inherited from the parent (as in Examples \ref{ex:moranProcess} and \ref{ex:wfExample}), or, alternatively, it may be determined by a player's spatial location \citep[see][]{mcavoy:JRSI:2015}. These details are encoded entirely in the update rule.

Although the framework we present here is clearly aimed at evolutionary \textit{games} used to describe natural selection, related processes that are not technically ``games" may also constitute stochastic selection processes. Evolutionary algorithms, for example, form an important subclass of stochastic selection processes. These algorithms seek to apply the principles of natural selection to solve search and optimization problems \citep{back:OUP:1996}. Evolutionary algorithms typically do not have population state spaces, which, in our context, means that $\mathcal{P}_{N}$ can be taken to be a singleton equipped with the trivial action of $\mathfrak{S}_{N}$. A popular type of evolutionary algorithm, known as a \textit{genetic} algorithm, involves representing the elements of the search space, i.e. the \textit{genomes} in $S$, as sequences of binary digits. Each genome is then assigned a fitness based on its viability as a solution to the problem at hand. (Unlike in biological populations, the fitness landscape, although complex, is inherently static and does not depend on the other members of the population.) The update step, which is commonly designed to mimic sexual reproduction in nature, involves a combination of selection, crossover, and mutation. A population of genomes is then repeatedly updated until a sufficiently fit genome appears. Despite the fact that biological reproduction generally involves either one (asexual) or two (sexual) parents, evolutionary algorithms have been simulated using \textit{many} parents \citep{chambers:CRC:1998}. Other components of the update step in some algorithms, such as stochastic universal sampling \citep{baker:GA:1987}, elitism \citep{baluja:MLP:1995}, and tournament selection \citep{poli:FGA:2005}, are all readily incorporated into our model of stochastic selection processes.

In Example \ref{ex:pcExample}, we saw an evolutionary process with an uncountably infinite state space. A state space of this sort arises naturally in the public goods game, for instance, where there is a continuous range of investment levels. This example also illustrates the more complicated ways in which strategy mutations can be incorporated into an evolutionary process. If a player has a strategy, $x\in\left[0,K\right]$ for some $K>0$, then it may be the case that this player is more likely to ``explore" strategies close to $x$ than he or she is to switch to strategies farther away. A truncated Gaussian random variable on $\left[0,K\right]$, whose variance is a measure of how venturesome a player is, captures this type of strategy exploration and is easily incorporated into the update rule of an evolutionary process. This type of mutation has appeared in the context of adaptive dynamics \citep{doebeli:S:2004} and, more recently, in a study of stochastic evolutionary branching \citep{wakano:G:2012}, but it has been largely ignored elsewhere in the literature on evolutionary game theory, where strategy mutations typically involve switching between two strategies or else are governed by a uniform random variable over the strategy space. Further studies of the dynamics of processes with these biologically-relevant mutations are certainly warranted.

Our framework makes no assumptions on the cardinality of $S$ and $\mathcal{P}_{N}$; all that is required is that these spaces be measurable (and that $\mathcal{P}_{N}$ be equipped with an action of $\mathfrak{S}_{N}$). Markov chains on continuous state spaces have unique stationary distributions under certain circumstances \citep[see][]{durrett:CUP:2009}, but, in the generality of this framework, it need not be the case that a stationary distribution is unique. Even if $\mathcal{S}$ is finite and there are nonzero strategy-mutation rates, the spatial structure of the population might be disconnected, resulting in multiple stationary distributions. Particular instances of stochastic selection processes may have the property that nonzero strategy-mutation rates imply that the Markov chain defined by the process is irreducible \citep{fudenberg:JET:2006,fudenberg:JET:2008,allen:JMB:2012}, but this phenomenon need not hold in general. Our goal here was not to study the dynamics of any particular subclass of stochastic selection processes, but rather to formalize \textit{what these processes are}.

Our general theory of stochastic selection processes provides a mathematical foundation for a broad class of processes used to describe evolution by means of natural selection in finite populations. Stochastic selection processes also provide a mathematical framework for processes with \textit{variable} population size, a topic that has received surprisingly little attention in the literature. Although many biological interactions have been modeled using classical games, the differences between stochastic games and stochastic selection processes illustrate a fundamental distinction between classical and evolutionary game theory. There is still a lot to discover about the dynamics of selection processes in finite populations (especially those with variable population size), and our hope is that this framework elucidates the roles of the components of processes based on natural selection and advances the effort to transform evolution into a mathematical theory.

\section*{Acknowledgments}

The author thanks Christoph Hauert for many helpful discussions and for carefully reading earlier versions of this manuscript. Financial support from the Natural Sciences and Engineering Research Council of Canada (NSERC) is gratefully acknowledged.

\bibliographystyle{plainnat}

\begin{thebibliography}{63}
\providecommand{\natexlab}[1]{#1}
\providecommand{\url}[1]{\texttt{#1}}
\expandafter\ifx\csname urlstyle\endcsname\relax
  \providecommand{\doi}[1]{doi: #1}\else
  \providecommand{\doi}{doi: \begingroup \urlstyle{rm}\Url}\fi

\bibitem[Allen and Tarnita(2012)]{allen:JMB:2012}
B.~Allen and C.~E. Tarnita.
\newblock Measures of success in a class of evolutionary models with fixed
  population size and structure.
\newblock \emph{Journal of Mathematical Biology}, 68\penalty0 (1-2):\penalty0
  109--143, Nov 2012.
\newblock \doi{10.1007/s00285-012-0622-x}.

\bibitem[Antal et~al.(2009)Antal, Ohtsuki, Wakeley, Taylor, and
  Nowak]{antal:PNAS:2009}
T.~Antal, H.~Ohtsuki, J.~Wakeley, P.~D. Taylor, and M.~A. Nowak.
\newblock Evolution of cooperation by phenotypic similarity.
\newblock \emph{Proceedings of the National Academy of Sciences}, 106\penalty0
  (21):\penalty0 8597--8600, Apr 2009.
\newblock \doi{10.1073/pnas.0902528106}.

\bibitem[Aumann(1987)]{aumann:E:1987}
R.~J. Aumann.
\newblock Correlated equilibrium as an expression of bayesian rationality.
\newblock \emph{Econometrica}, 55\penalty0 (1):\penalty0 1, Jan 1987.
\newblock \doi{10.2307/1911154}.

\bibitem[B\"{a}ck(1996)]{back:OUP:1996}
T.~B\"{a}ck.
\newblock \emph{Evolutionary Algorithms in Theory and Practice: Evolution
  Strategies, Evolutionary Programming, Genetic Algorithms}.
\newblock Oxford University Press, Oxford, UK, 1996.
\newblock ISBN 0-19-509971-0.

\bibitem[Baker(1987)]{baker:GA:1987}
J.~E. Baker.
\newblock Reducing bias and inefficiency in the selection algorithm.
\newblock In \emph{Proceedings of the Second International Conference on
  Genetic Algorithms on Genetic Algorithms and Their Application}, pages
  14--21, Hillsdale, NJ, USA, 1987. L. Erlbaum Associates Inc.
\newblock ISBN 0-8058-0158-8.

\bibitem[Baluja and Caruana(1995)]{baluja:MLP:1995}
S.~Baluja and R.~Caruana.
\newblock Removing the genetics from the standard genetic algorithm.
\newblock In \emph{Machine Learning Proceedings 1995}, pages 38--46. Elsevier,
  1995.
\newblock \doi{10.1016/b978-1-55860-377-6.50014-1}.

\bibitem[Chambers(1998)]{chambers:CRC:1998}
L.~Chambers, editor.
\newblock \emph{Practical Handbook of Genetic Algorithms}.
\newblock {CRC} Press, Dec 1998.
\newblock \doi{10.1201/9781420050080}.

\bibitem[D{\'{e}}barre et~al.(2014)D{\'{e}}barre, Hauert, and
  Doebeli]{debarre:NC:2014}
F.~D{\'{e}}barre, C.~Hauert, and M.~Doebeli.
\newblock Social evolution in structured populations.
\newblock \emph{Nature Communications}, 5, Mar 2014.
\newblock \doi{10.1038/ncomms4409}.

\bibitem[Der et~al.(2011)Der, Epstein, and Plotkin]{der:TPB:2011}
R.~Der, C.~L. Epstein, and J.~B. Plotkin.
\newblock Generalized population models and the nature of genetic drift.
\newblock \emph{Theoretical Population Biology}, 80\penalty0 (2):\penalty0
  80--99, Sep 2011.
\newblock \doi{10.1016/j.tpb.2011.06.004}.

\bibitem[Doebeli et~al.(2004)Doebeli, Hauert, and Killingback]{doebeli:S:2004}
M.~Doebeli, C.~Hauert, and T.~Killingback.
\newblock {The Evolutionary Origin of Cooperators and Defectors}.
\newblock \emph{Science}, 306\penalty0 (5697):\penalty0 859--862, Oct 2004.
\newblock \doi{10.1126/science.1101456}.

\bibitem[Dugatkin(2000)]{dugatkin:OUP:2000}
L.~A. Dugatkin.
\newblock \emph{Game Theory and Animal Behavior}.
\newblock Oxford University Press, 2000.

\bibitem[Durrett(2002)]{durrett:SNY:2002}
R.~Durrett.
\newblock \emph{Probability Models for {DNA} Sequence Evolution}.
\newblock Springer New York, 2002.
\newblock \doi{10.1007/978-1-4757-6285-3}.

\bibitem[Durrett(2009)]{durrett:CUP:2009}
R.~Durrett.
\newblock \emph{Probability: Theory and Examples}.
\newblock Cambridge University Press, 2009.
\newblock \doi{10.1017/cbo9780511779398}.

\bibitem[Ellison(1993)]{ellison:E:1993}
G.~Ellison.
\newblock Learning, local interaction, and coordination.
\newblock \emph{Econometrica}, 61\penalty0 (5):\penalty0 1047, Sep 1993.
\newblock \doi{10.2307/2951493}.

\bibitem[Ewens(2004)]{ewens:S:2004}
W.~J. Ewens.
\newblock \emph{Mathematical Population Genetics}.
\newblock Springer New York, 2004.
\newblock \doi{10.1007/978-0-387-21822-9}.

\bibitem[Frucht(1939)]{frucht:CM:1939}
R.~Frucht.
\newblock Herstellung von graphen mit vorgegebener abstrakter gruppe.
\newblock \emph{Compositio Mathematica}, 6:\penalty0 239--250, 1939.

\bibitem[Fudenberg and Imhof(2006)]{fudenberg:JET:2006}
D.~Fudenberg and L.~A. Imhof.
\newblock Imitation processes with small mutations.
\newblock \emph{Journal of Economic Theory}, 131\penalty0 (1):\penalty0
  251--262, Nov 2006.
\newblock \doi{10.1016/j.jet.2005.04.006}.

\bibitem[Fudenberg and Imhof(2008)]{fudenberg:JET:2008}
D.~Fudenberg and L.~A. Imhof.
\newblock Monotone imitation dynamics in large populations.
\newblock \emph{Journal of Economic Theory}, 140\penalty0 (1):\penalty0
  229--245, May 2008.
\newblock \doi{10.1016/j.jet.2007.08.002}.

\bibitem[Fudenberg and Tirole(1991)]{fudenberg:MIT:1991}
D.~Fudenberg and J.~Tirole.
\newblock \emph{Game Theory}.
\newblock The MIT Press, 1991.

\bibitem[Hart and Mas-Colell(2000)]{hart:E:2000}
S.~Hart and A.~Mas-Colell.
\newblock A simple adaptive procedure leading to correlated equilibrium.
\newblock \emph{Econometrica}, 68\penalty0 (5):\penalty0 1127--1150, Sep 2000.
\newblock \doi{10.1111/1468-0262.00153}.

\bibitem[Hauert and Imhof(2012)]{hauert:JTB:2012}
C.~Hauert and L.~Imhof.
\newblock Evolutionary games in deme structured, finite populations.
\newblock \emph{Journal of Theoretical Biology}, 299:\penalty0 106--112, Apr
  2012.
\newblock \doi{10.1016/j.jtbi.2011.06.010}.

\bibitem[Hauert and Schuster(1997)]{hauert:PRSLB:1997}
C.~Hauert and H.~G. Schuster.
\newblock Effects of increasing the number of players and memory size in the
  iterated prisoner{\textquoteright}s dilemma: a numerical approach.
\newblock \emph{Proceedings of the Royal Society of London B: Biological
  Sciences}, 264\penalty0 (1381):\penalty0 513--519, 1997.
\newblock \doi{10.1098/rspb.1997.0073}.

\bibitem[Hofbauer and Sigmund(1998)]{hofbauer:CUP:1998}
J.~Hofbauer and K.~Sigmund.
\newblock \emph{Evolutionary Games and Population Dynamics}.
\newblock Cambridge University Press, 1998.
\newblock \doi{10.1017/cbo9781139173179}.

\bibitem[Imhof and Nowak(2006)]{imhof:JMB:2006}
L.~A. Imhof and M.~A. Nowak.
\newblock Evolutionary game dynamics in a wright-fisher process.
\newblock \emph{Journal of Mathematical Biology}, 52\penalty0 (5):\penalty0
  667--681, Feb 2006.
\newblock \doi{10.1007/s00285-005-0369-8}.

\bibitem[Kingman(1980)]{kingman:SIAM:1980}
J.~F.~C. Kingman.
\newblock \emph{Mathematics of Genetic Diversity}.
\newblock Society for Industrial and Applied Mathematics, Jan 1980.
\newblock \doi{10.1137/1.9781611970357}.

\bibitem[Knapp(2006)]{knapp:BB:2006}
A.~W. Knapp.
\newblock \emph{Basic Algebra}.
\newblock Birkhäuser Boston, 2006.
\newblock \doi{10.1007/978-0-8176-4529-8}.

\bibitem[Lieberman et~al.(2005)Lieberman, Hauert, and
  Nowak]{lieberman:Nature:2005}
E.~Lieberman, C.~Hauert, and M.~A. Nowak.
\newblock Evolutionary dynamics on graphs.
\newblock \emph{Nature}, 433\penalty0 (7023):\penalty0 312--316, Jan 2005.
\newblock \doi{10.1038/nature03204}.

\bibitem[Maciejewski et~al.(2014)Maciejewski, Fu, and
  Hauert]{maciejewski:PLoSCB:2014}
W.~Maciejewski, F.~Fu, and C.~Hauert.
\newblock Evolutionary game dynamics in populations with heterogenous
  structures.
\newblock \emph{PLoS Computational Biology}, 10\penalty0 (4):\penalty0
  e1003567, Apr 2014.
\newblock \doi{10.1371/journal.pcbi.1003567}.

\bibitem[Maynard~Smith(1982)]{maynardsmith:CUP:1982}
J.~Maynard~Smith.
\newblock \emph{Evolution and the Theory of Games}.
\newblock Cambridge University Press, 1982.
\newblock \doi{10.1017/cbo9780511806292}.

\bibitem[McAvoy and Hauert(2015{\natexlab{a}})]{mcavoy:JMB:2015}
A.~McAvoy and C.~Hauert.
\newblock Structure coefficients and strategy selection in multiplayer games.
\newblock \emph{Journal of Mathematical Biology}, Apr 2015{\natexlab{a}}.
\newblock \doi{10.1007/s00285-015-0882-3}.

\bibitem[McAvoy and Hauert(2015{\natexlab{b}})]{mcavoy:JRSI:2015}
A.~McAvoy and C.~Hauert.
\newblock Structural symmetry in evolutionary games.
\newblock \emph{Journal of The Royal Society Interface}, 12\penalty0
  (111):\penalty0 20150420, Sep 2015{\natexlab{b}}.
\newblock \doi{10.1098/rsif.2015.0420}.

\bibitem[McAvoy and Hauert(2015{\natexlab{c}})]{mcavoy:PLOSCB:2015}
A.~McAvoy and C.~Hauert.
\newblock Asymmetric evolutionary games.
\newblock \emph{{PLOS} Computational Biology}, 11\penalty0 (8):\penalty0
  e1004349, Aug 2015{\natexlab{c}}.
\newblock \doi{10.1371/journal.pcbi.1004349}.

\bibitem[McMillan(2001)]{mcmillan:R:2001}
J.~McMillan.
\newblock \emph{Game Theory in International Economics}.
\newblock Routledge, Dec 2001.
\newblock \doi{10.4324/9781315014562}.

\bibitem[Mertens et~al.(2015)Mertens, Sorin, and Zamir]{mertens:CUP:2015}
J.-F. Mertens, S.~Sorin, and S.~Zamir.
\newblock \emph{Repeated Games}.
\newblock Cambridge University Press, 2015.
\newblock \doi{10.1017/cbo9781139343275}.

\bibitem[Moran(1958)]{moran:MPCPS:1958}
P.~A.~P. Moran.
\newblock Random processes in genetics.
\newblock \emph{Mathematical Proceedings of the Cambridge Philosophical
  Society}, 54\penalty0 (01):\penalty0 60, Jan 1958.
\newblock \doi{10.1017/s0305004100033193}.

\bibitem[Neyman(2003)]{neyman:FMCSG:2003}
A.~Neyman.
\newblock From markov chains to stochastic games.
\newblock In \emph{Stochastic Games and Applications}, pages 9--25. Springer
  Netherlands, 2003.
\newblock \doi{10.1007/978-94-010-0189-2$\_$2}.

\bibitem[Neyman and Sorin(2003)]{neyman:S:2003}
A.~Neyman and S.~Sorin, editors.
\newblock \emph{Stochastic Games and Applications}.
\newblock Springer Netherlands, 2003.
\newblock \doi{10.1007/978-94-010-0189-2}.

\bibitem[Nowak(2006)]{nowak:BP:2006}
M.~A. Nowak.
\newblock \emph{Evolutionary Dynamics: Exploring the Equations of Life}.
\newblock Belknap Press, 2006.

\bibitem[Nowak et~al.(2004)Nowak, Sasaki, Taylor, and
  Fudenberg]{nowak:Nature:2004}
M.~A. Nowak, A.~Sasaki, C.~Taylor, and D.~Fudenberg.
\newblock Emergence of cooperation and evolutionary stability in finite
  populations.
\newblock \emph{Nature}, 428\penalty0 (6983):\penalty0 646--650, Apr 2004.
\newblock \doi{10.1038/nature02414}.

\bibitem[Nowak et~al.(2009)Nowak, Tarnita, and Antal]{nowak:PTRSB:2009}
M.~A. Nowak, C.~E. Tarnita, and T.~Antal.
\newblock Evolutionary dynamics in structured populations.
\newblock \emph{Philosophical Transactions of the Royal Society B: Biological
  Sciences}, 365\penalty0 (1537):\penalty0 19--30, Nov 2009.
\newblock \doi{10.1098/rstb.2009.0215}.

\bibitem[Ohtsuki and Nowak(2006)]{ohtsuki:JTB:2006}
H.~Ohtsuki and M.~A. Nowak.
\newblock The replicator equation on graphs.
\newblock \emph{Journal of Theoretical Biology}, 243\penalty0 (1):\penalty0
  86--97, Nov 2006.
\newblock \doi{10.1016/j.jtbi.2006.06.004}.

\bibitem[Ohtsuki et~al.(2006)Ohtsuki, Hauert, Lieberman, and
  Nowak]{ohtsuki:Nature:2006}
H.~Ohtsuki, C.~Hauert, E.~Lieberman, and M.~A. Nowak.
\newblock A simple rule for the evolution of cooperation on graphs and social
  networks.
\newblock \emph{Nature}, 441\penalty0 (7092):\penalty0 502--505, May 2006.
\newblock \doi{10.1038/nature04605}.

\bibitem[Ohtsuki et~al.(2007{\natexlab{a}})Ohtsuki, Nowak, and
  Pacheco]{ohtsuki:PRL:2007}
H.~Ohtsuki, M.~A. Nowak, and J.~M. Pacheco.
\newblock Breaking the symmetry between interaction and replacement in
  evolutionary dynamics on graphs.
\newblock \emph{Physical Review Letters}, 98\penalty0 (10), Mar
  2007{\natexlab{a}}.
\newblock \doi{10.1103/physrevlett.98.108106}.

\bibitem[Ohtsuki et~al.(2007{\natexlab{b}})Ohtsuki, Pacheco, and
  Nowak]{ohtsuki:JTB:2007}
H.~Ohtsuki, J.~M. Pacheco, and M.~A. Nowak.
\newblock Evolutionary graph theory: Breaking the symmetry between interaction
  and replacement.
\newblock \emph{Journal of Theoretical Biology}, 246\penalty0 (4):\penalty0
  681--694, Jun 2007{\natexlab{b}}.
\newblock \doi{10.1016/j.jtbi.2007.01.024}.

\bibitem[Pacheco et~al.(2009)Pacheco, Pinheiro, and
  Santos]{pacheco:PLoSCB:2009}
J.~M. Pacheco, F.~L. Pinheiro, and F.~C. Santos.
\newblock Population structure induces a symmetry breaking favoring the
  emergence of cooperation.
\newblock \emph{{PLoS} Computational Biology}, 5\penalty0 (12):\penalty0
  e1000596, Dec 2009.
\newblock \doi{10.1371/journal.pcbi.1000596}.

\bibitem[Poli(2005)]{poli:FGA:2005}
R.~Poli.
\newblock Tournament selection, iterated coupon-collection problem, and
  backward-chaining evolutionary algorithms.
\newblock In \emph{Foundations of Genetic Algorithms}, pages 132--155. Springer
  Science + Business Media, 2005.
\newblock \doi{10.1007/11513575$\_$8}.

\bibitem[Press and Dyson(2012)]{press:PNAS:2012}
W.~H. Press and F.~J. Dyson.
\newblock Iterated prisoner's dilemma contains strategies that dominate any
  evolutionary opponent.
\newblock \emph{Proceedings of the National Academy of Sciences}, 109\penalty0
  (26):\penalty0 10409--10413, May 2012.
\newblock \doi{10.1073/pnas.1206569109}.

\bibitem[Puterman(1994)]{puterman:JWS:1994}
M.~L. Puterman, editor.
\newblock \emph{Markov Decision Processes}.
\newblock John Wiley {\&} Sons, Inc., Apr 1994.
\newblock \doi{10.1002/9780470316887}.

\bibitem[Rand et~al.(2014)Rand, Nowak, Fowler, and Christakis]{rand:PNAS:2014}
D.~G. Rand, M.~A. Nowak, J.~H. Fowler, and N.~A. Christakis.
\newblock Static network structure can stabilize human cooperation.
\newblock \emph{Proceedings of the National Academy of Sciences}, 111\penalty0
  (48):\penalty0 17093--17098, Nov 2014.
\newblock \doi{10.1073/pnas.1400406111}.

\bibitem[Shapley(1953)]{shapley:PNAS:1953}
L.~S. Shapley.
\newblock Stochastic games.
\newblock \emph{Proceedings of the National Academy of Sciences}, 39\penalty0
  (10):\penalty0 1095--1100, Oct 1953.
\newblock \doi{10.1073/pnas.39.10.1095}.

\bibitem[Szab\'{o} and F\'{a}th(2007)]{szabo:PR:2007}
G.~Szab\'{o} and G.~F\'{a}th.
\newblock Evolutionary games on graphs.
\newblock \emph{Physics Reports}, 446\penalty0 (4-6):\penalty0 97--216, Jul
  2007.
\newblock \doi{10.1016/j.physrep.2007.04.004}.

\bibitem[Szab{\'{o}} and T{\H{o}}ke(1998)]{szabo:PRE:1998}
G.~Szab{\'{o}} and C.~T{\H{o}}ke.
\newblock Evolutionary prisoner's dilemma game on a square lattice.
\newblock \emph{Physical Review E}, 58\penalty0 (1):\penalty0 69--73, Jul 1998.
\newblock \doi{10.1103/physreve.58.69}.

\bibitem[Tarnita et~al.(2009)Tarnita, Antal, Ohtsuki, and
  Nowak]{tarnita:PNAS:2009}
C.~E. Tarnita, T.~Antal, H.~Ohtsuki, and M.~A. Nowak.
\newblock Evolutionary dynamics in set structured populations.
\newblock \emph{Proceedings of the National Academy of Sciences}, 106\penalty0
  (21):\penalty0 8601--8604, May 2009.
\newblock \doi{10.1073/pnas.0903019106}.

\bibitem[Taylor et~al.(2004)Taylor, Fudenberg, Sasaki, and
  Nowak]{taylor:BMB:2004}
C.~Taylor, D.~Fudenberg, A.~Sasaki, and M.~A. Nowak.
\newblock Evolutionary game dynamics in finite populations.
\newblock \emph{Bulletin of Mathematical Biology}, 66\penalty0 (6):\penalty0
  1621--1644, Nov 2004.
\newblock \doi{10.1016/j.bulm.2004.03.004}.

\bibitem[Taylor and Jonker(1978)]{taylor:MB:1978}
P.~D. Taylor and L.~B. Jonker.
\newblock Evolutionary stable strategies and game dynamics.
\newblock \emph{Mathematical Biosciences}, 40\penalty0 (1-2):\penalty0
  145--156, jul 1978.
\newblock \doi{10.1016/0025-5564(78)90077-9}.

\bibitem[Taylor et~al.(2001)Taylor, Irwin, and Day]{taylor:Selection:2000}
P.~D. Taylor, A.~J. Irwin, and T.~Day.
\newblock Inclusive fitness in finite deme-structured and stepping-stone
  populations.
\newblock \emph{Selection}, 1\penalty0 (1):\penalty0 153--164, Jan 2001.
\newblock \doi{10.1556/select.1.2000.1-3.15}.

\bibitem[Taylor et~al.(2007)Taylor, Day, and Wild]{taylor:Nature:2007}
P.~D. Taylor, T.~Day, and G.~Wild.
\newblock Evolution of cooperation in a finite homogeneous graph.
\newblock \emph{Nature}, 447\penalty0 (7143):\penalty0 469--472, May 2007.
\newblock \doi{10.1038/nature05784}.

\bibitem[Traulsen et~al.(2007)Traulsen, Pacheco, and Nowak]{traulsen:JTB:2007}
A.~Traulsen, J.~M. Pacheco, and M.~A. Nowak.
\newblock Pairwise comparison and selection temperature in evolutionary game
  dynamics.
\newblock \emph{Journal of Theoretical Biology}, 246\penalty0 (3):\penalty0
  522--529, Jun 2007.
\newblock \doi{10.1016/j.jtbi.2007.01.002}.

\bibitem[Traulsen et~al.(2009)Traulsen, Hauert, De~Silva, Nowak, and
  Sigmund]{traulsen:PNAS:2009}
A.~Traulsen, C.~Hauert, H.~De~Silva, M.~A. Nowak, and K.~Sigmund.
\newblock Exploration dynamics in evolutionary games.
\newblock \emph{Proceedings of the National Academy of Sciences}, 106\penalty0
  (3):\penalty0 709--712, Jan 2009.
\newblock \doi{10.1073/pnas.0808450106}.

\bibitem[Wage(2010)]{wage:thesis:2010}
Nicholas Wage.
\newblock Evolutionary games on structured populations under weak selection.
\newblock Harvard senior thesis, 2010.

\bibitem[Wakano and Iwasa(2012)]{wakano:G:2012}
J.~Y. Wakano and Y.~Iwasa.
\newblock Evolutionary branching in a finite population: Deterministic
  branching vs. stochastic branching.
\newblock \emph{Genetics}, 193\penalty0 (1):\penalty0 229--241, Oct 2012.
\newblock \doi{10.1534/genetics.112.144980}.

\bibitem[Wakeley and Takahashi(2004)]{wakeley:TPB:2004}
J.~Wakeley and T.~Takahashi.
\newblock The many-demes limit for selection and drift in a subdivided
  population.
\newblock \emph{Theoretical Population Biology}, 66\penalty0 (2):\penalty0
  83--91, Sep 2004.
\newblock \doi{10.1016/j.tpb.2004.04.005}.

\bibitem[Wardil and Hauert(2014)]{wardil:SR:2014}
L.~Wardil and C.~Hauert.
\newblock Origin and structure of dynamic cooperative networks.
\newblock \emph{Scientific Reports}, 4, Jul 2014.
\newblock \doi{10.1038/srep05725}.

\end{thebibliography}

\end{document}